\documentclass[11pt]{article}

\usepackage{mathtools}
\usepackage{amsfonts}
\usepackage{mathrsfs}
\usepackage{amssymb}

\usepackage{amsthm}

\usepackage{xcolor}
\usepackage{xspace}
\usepackage{paralist}
\usepackage{fullpage}
\usepackage{times}
\usepackage[pdfpagelabels,linkcolor=blue,urlcolor=blue,pdfstartview=FitH,pdfpagemode=None]{hyperref}
\usepackage{bm}
\usepackage{graphicx}
\usepackage{subfigure}
\usepackage[margin=1in]{geometry}
\usepackage{xargs}
\usepackage{tikz}
\usepackage{hyperref}
\usepackage[capitalize]{cleveref}

\usepackage[suppress]{color-edits}
\addauthor{bk}{blue}
\addauthor{ag}{red}

\newtheorem{theorem}{Theorem}[section]

\newtheorem{corollary}[theorem]{Corollary}

\newtheorem{lemma}[theorem]{Lemma}

\theoremstyle{definition}

\newtheorem{definition}[theorem]{Definition}

\newtheorem{remark}[theorem]{Remark}

\def\Snospace~{\S{}}

\newcommand{\eg}{{\it e.g.}\xspace}
\newcommand{\ie}{{\it i.e.}\xspace}
\newcommand{\BEAS}{\begin{eqnarray*}}
\newcommand{\EEAS}{\end{eqnarray*}}
\newcommand{\BEA}{\begin{eqnarray}}
\newcommand{\EEA}{\end{eqnarray}}
\newcommand{\BIT}{\begin{itemize}}
\newcommand{\EIT}{\end{itemize}}
\newcommand{\BEQ}{\begin{equation}}
\newcommand{\EEQ}{\end{equation}}
\newcommand{\BEQS}{\begin{equation*}}
\newcommand{\EEQS}{\end{equation*}}
\newcommand{\BNUM}{\begin{enumerate}}
\newcommand{\ENUM}{\end{enumerate}}

\newcommand{\reals}{{\mathbb{R}}}

\DeclareMathOperator{\argmax}{arg\,max}
\DeclareMathOperator{\argmin}{arg\,min}

\newcommand{\idp}{\epsilon}
\newcommand{\idpa}{{\idp_a}} 
\newcommand{\idpi}{{\idp_i}} 
\newcommand{\idpv}{{\bm{\idp}}}
\newcommand{\ndp}{\nu}
\newcommand{\ndpa}{\nu_a}

\newcommand{\ba}{x_a} 
\newcommand{\dba}{x_a} 
\newcommand{\fdb}{{\bf x}} 
\newcommand{\pdb}{{\bf y}} 

\newcommand{\dbi}{{x_i}} 
\newcommand{\adj}[1]{{#1}'} 
\newcommand{\adb}{{\adj{\fdb}}} 
\newcommand{\bg}{x} 

\newcommand{\m}{p} 
\newcommand{\jd}{\mu}  
\newcommand{\cdu}{\jd^1}  
\newcommand{\cdl}{\jd^0}  

\newcommand{\bn}{\{0,1\}^n}

\newcommand{\vals}{X}
\newcommand{\val}{z}

\newcommand{\M}{{\mathcal{M}}}

\newcommand{\infadd}{{\xi}}

\newcommand{\infmul}{{\gamma}}
\newcommand{\infmulmtx}{\Gamma}
\newcommand{\eps}{\varepsilon}
\newcommand{\outcomes}{{\mathcal{O}}}
\newcommand{\outc}{o}
\DeclareMathOperator{\supp}{supp}
\newcommand{\estim}{\kappa}

\newcommand{\lip}{{\rho}}
\newcommand{\cd}[1]{{\pi^{#1}}}
\newcommand{\invmtx}{\Phi}

\newcommand{\invmulmtx}{\Phi}
\newcommand{\invmulentry}{\phi}
\newcommand{\cmtx}{\Upsilon}
\newcommand{\expect}{{\mathbb{E}}}
\newcommand{\given}{ \mid }
\newcommand{\avg}{\tau}
\newcommand{\avgi}{\tau_i}

\newcommand{\avgf}{\avg f}
\newcommand{\avgg}{\avg g}

\newcommand{\omf}{1\!-\!f}
\newcommand{\prat}{{\rho}}
\newcommand{\jds}{{\mathscr D}}
\newcommand{\graphs}{{\mathscr G}}

\newcommand{\transpose}{\intercal}

\newcommand{\state}{{\bm{\sigma}}}
\newcommand{\extfld}{{\bm{h}}}

\newcommand{\gavg}[1]{{\left\langle #1 \right\rangle}}

\newcommand{\infpriv}{inferential privacy\xspace}
\newcommand{\NS}{{\mathcal{NS}}}

\newcommand{\jj}{J}

\newif\ifextabs
\extabstrue

\newif\ify
\ytrue

\newif\ifproof

\ifextabs
\prooffalse
\else
\prooftrue
\fi

\ifextabs
\usepackage[numbers]{natbib}
\else
\usepackage{natbib}
\fi
\begin{document}
\title{Inferential Privacy Guarantees for Differentially Private Mechanisms}
\author{
  Arpita Ghosh\thanks{%
    Cornell University, Ithaca, NY, USA. 
    \texttt{arpitaghosh@cornell.edu}
  }
 \and 
 Robert Kleinberg\thanks{%
   Cornell University, Ithaca, NY, USA, and 
   Microsoft Research New England, Cambridge, MA, USA.
   \texttt{robert.kleinberg@cornell.edu}
 }
}
\date{}

\begin{titlepage}
\maketitle
\thispagestyle{empty}

\begin{abstract}
The correlations and network structure amongst individuals in datasets today---whether explicitly articulated, or deduced from biological or behavioral connections---pose new issues around privacy guarantees, because of inferences that can be made about one individual from another's data. This motivates quantifying privacy in networked contexts in terms of `inferential privacy'---which measures the change in beliefs about an individual's data from the result of a computation---as originally proposed by Dalenius in the 1970's. Inferential privacy is implied by differential privacy when data are independent, but can be much worse when data are correlated; indeed, simple examples, as well as a general impossibility theorem of Dwork and Naor, preclude the possibility of achieving non-trivial inferential privacy when the adversary can have arbitrary auxiliary information. In this paper, we ask how differential privacy guarantees translate to guarantees on inferential privacy in networked contexts: specifically, under what limitations on the adversary's information about correlations, modeled as a prior distribution over datasets, can we deduce an inferential guarantee from a differential one? 

We prove two main results. The first result pertains to distributions that satisfy a natural positive-affiliation condition, and gives an upper bound on the inferential privacy guarantee for any differentially private mechanism. This upper bound is matched by a simple mechanism that adds Laplace noise to the sum of the data. The second result pertains to distributions that have weak correlations, defined in terms of a suitable ``influence matrix''. The result provides an upper bound for inferential privacy in terms of the differential privacy parameter and the spectral norm of this matrix.
\end{abstract}
\end{titlepage}

\ify
\section{Introduction}
\label{sec:intro}

Privacy has always been a central issue in the discourse surrounding the collection and use of personal data. As the nature of data collected online grows richer, however, fundamentally new privacy issues emerge. In a thought-provoking piece entitled ``Networked Rights and Networked Harms'' \cite{LB15}, the sociologists Karen Levy and danah boyd argue that the `networks' surrounding data today---whether articulated (as in explicitly declared friendships on social networks), behavioral (as in connections inferred from observed behavior), or biological (as in genetic databases)---raise conceptually new questions that current privacy law and policy cannot address. \citeauthor{LB15} present case studies to demonstrate how the current individual-centric legal frameworks for privacy do not provide a means to account for the networked contexts now surrounding personal data.

An analogous question arises on the formal front. One of computer science's fundamental contributions to the public debate about private data---most prominently via the literature on {\em differential privacy\footnote{Differential privacy, which measures privacy via the relative amount of new information disclosed about an individual's data by her participation in a dataset, has emerged as the primary theoretical framework for quantifying privacy loss. }}~\cite{DPbook}---has been to  provide a means to {\em measure} privacy loss, which enables evaluating the privacy implications of proposed data analyses and disclosures in quantitative terms. However, differential privacy focuses on the privacy loss to an individual by {\em her} contribution to a dataset, and therefore---by design---does not capture all of the privacy losses from inferences that could be made about one person's data due to its correlations with {\em other} data in networked contexts. For instance, the privacy implications of a database such as 23andme for one individual depend not just on that person's own data and the computation performed, but also on her siblings' data.

In this paper, we look to understand the implications of such `networked' data for formal privacy guarantees. How much can be learnt about a single individual from the result of a computation on  correlated data, and how does this relate to the differential privacy guarantee of the computation? 
\paragraph{Inferential privacy.} 
A natural way of assessing whether a mechanism $\M$ protects the privacy of an individual is to ask, ``Is it possible that someone,
after observing the mechanism's output, will learn a lot about the individual's private data?'' In other words, what is the {\em inferential privacy}---the largest possible ratio between the posterior and prior beliefs about an individual's data after observing the result of a computation on the database? (This quantity is identical to the differential privacy parameter of the mechanism when individuals' data are independent; see \S\ref{sec:def-ndp} and \cite{KS14}.) 

The inferential privacy guarantee will depend, of course, on both the nature of the correlations in the database and on the precise mechanism used to perform the computation. Instead of seeking to {\em design} algorithms that achieve a particular inferential privacy guarantee---which would necessitate choosing a particular computational objective and correlation structure---we instead seek to {\em analyze} the inferential privacy guarantees provided by differentially private algorithms. Specifically, we ask the following question: consider the class of all mechanisms providing a certain differential privacy guarantee, say $\eps$. What is the worst possible inferential privacy guarantee for a mechanism in this class?

This question is pertinent to a policy-maker who can prescribe that analysts provide some degree of differential privacy to individuals while releasing their results, but cannot control how---\ie, using what specific algorithm---the analyst will provide this guarantee. In other words, rather than an algorithm designer who wants to design an inferential privacy-preserving algorithm (for a particular scenario), this question adopts the perspective of a policy-maker who can set privacy standards that analysts must obey, but is agnostic to the analysts' computational objectives. We choose the differential privacy guarantee as our measure of privacy for many reasons: it is, at present, the only widely-agreed-upon privacy guarantee known to provide strong protections even against arbitrary side information; there is a vast toolbox of differentially private algorithms and a well-understood set of composition rules for combining them to yield new ones; finally, differential privacy is now beginning to make its way into policy and legal frameworks as a potential means for quantifying privacy loss. 

Measuring privacy loss via inferential privacy
formalizes Dalenius's~\cite{Dal77} desideratum that ``access to a statistical database should not enable one to learn anything about an individual that could not be learned without access''. 
While it is well known\footnote{see, \eg, \cite{DworkNaor,DPbook}} that non-trivial inferential privacy guarantees are incompatible with non-trivial utility guarantees in the presence of arbitrary auxiliary information, our primary contribution is modeling and quantifying {\em what} degree of inferential privacy is in fact achievable under a {\em particular} form of auxiliary information, such as that resulting from a known correlation structure or a limited set of such structures. For example, as noted earlier, if the individuals' rows in the database are conditionally independent given the adversary's auxiliary information, then the inferential privacy guarantee for any individual collapses to her differential privacy guarantee. At the other extreme, when all individuals' data are perfectly correlated, the inferential privacy parameter can exceed the differential privacy parameter by a factor of $n$ (the number of individuals in the database) as we will see below. What happens for correlations that lie somewhere in between these two extremes? Do product distributions belong to a broader class of distributions with {\em benign correlations} which ensure that an individual's \infpriv is not much worse than her differential privacy? \bkedit{A key 
contribution of our paper (\autoref{t-bm})
answers this question affirmatively while linking it to 
a well-known sufficient condition for `correlation decay'
in mathematical physics.}

\paragraph{Correlations in networked datasets and their privacy consequences.} 
We start with a caricature example to begin exploring how one might address these questions in a formal framework. Consider a database which contains an individual Athena and her (hypothetical) identical twin Adina, who is so identical to Athena that the rows in the database corresponding to Athena and Adina are identical in (the databases corresponding to) every possible state of the world. 
A differential privacy guarantee of $\idp$ to all database participants translates to an inferential privacy guarantee of only $2\idp$ to Athena (and her twin), since the ``neighboring'' database where Athena and Adina are different simply cannot exist.\footnote{Differential privacy guarantees that the probability of an outcome $\outc$ changes by at most a factor $e^{\idp}$ amongst databases at Hamming distance one, so that if $\fdb_1, \fdb_2$, and $\fdb_3$ denote the databases where the bits of Athena and Adina are $(0,0)$, $(1,0)$ and $(1,1)$ respectively, differential privacy guarantees that
$
\Pr(\outc|\fdb_1) \leq e^{\idp}\cdot \Pr(\outc|\fdb_2) \leq e^{2 \idp}\cdot \Pr(\outc|\fdb_3).
$
From here, a simple calculation using Bayes' Law---see equation~\eqref{eq:def-ndp.2} in \cref{sec:def-ndp}---implies
that:
$
\frac{\Pr(Athena=1|o)/\Pr(Athena=0|o)}{\Pr(Athena=1)/\Pr(Athena=0)} \leq e^{2\idp}, 
$
so that the inferential privacy guarantee is $2\idp$.} 

The erosion of Athena's privacy becomes even more extreme if the database contains $n>2$ individuals and they are all clones of Athena; a generalization of the preceding calculation now shows that the inferential privacy parameter is $n\eps$. However, in reality one is unlikely to participate in a database with many identical clones of oneself. Instead, it is interesting to consider cases with non-extreme correlations. For example, suppose now that the database contains data from Zeus and all of his descendants, and that every child's bit matches the parent's bit with probability $p > \frac12$. The degree of privacy afforded to Zeus now depends on many aspects of the model: the strength of the correlation ($p$), the number of individuals in the database ($n$), and structural properties of the tree of family relationships---its branching factor and depth, for instance. Which of these parameters contribute most crucially to inferential privacy? Is Zeus more likely to be implicated by his strong correlation with a few close relatives, or by a diffuse ``dragnet'' of correlations with his distant offspring?

In general, of course, networked databases, and the corresponding inferential privacy guarantees, do not come with as neat or convenient a correlation structure as in this example. In full generality, we can represent the idea of networked similarity via a joint distribution on databases that gives the prior probability of each particular combination of bits. So, for example, a world where all individuals in the database are ``twins'' would correspond to a joint distribution which has non-zero probability only on the all-zeros and all-ones databases, whereas a world where everyone's data is independent has multiplicative probabilities for each database. 

Such a model of correlations allows capturing a rich variety of networked contexts: in addition to situations where a single database contains sensitive information about $n$ individuals whose data have known correlations, it also captures the situation---perhaps closest to reality---where there are multiple databases to which multiple individuals contribute different (but correlated) pieces of information. In this latter interpretation, an inferential privacy guarantee limits the amount that an adversary may learn about {\em one} individual's contribution to {\em one} database, despite the correlations both across individuals and between a single individual's contributions to different databases.\footnote{We are grateful to Kobbi Nissim for suggesting this interpretation of our model.}

\ifextabs\else
In general, of course, networked databases, and the corresponding networked privacy guarantees, do not come with as neat or convenient a correlation structure as in this example: real databases with network structure---whether biological, behavioral or social---contain (i) more than one individual with (ii) degrees of `similarity' to each other that lie somewhere between these extremes of being either unrelated or identical. In other words, you and I and everyone we know has something like, but not quite, a twin---and many such semi-twins, in fact---in (and because of) the datasets containing information about us today. How much privacy do we really get from a given computation---and how can one begin to formalize privacy guarantees in such `networked' databases in a manner that yields analytical insight?

Now suppose that instead of containing the hypothetical twin Adina in addition to Athena, the database contains Athena's (mythical but non-hypothetical) relative, Zeus---now we have the feature of non-extreme correlation between individuals. We quantify the relation that Zeus and Athena are not identical, but still similar enough to each other to cause privacy concern, by saying that Athena and Zeus's bits are {\em the same with some probability} $p$ between $\frac{1}{2}$ and $1$---the extreme value of $\frac{1}{2}$ corresponds to independence
and the extreme of $1$ to perfect correlation. Generalizing still further, we can suppose that the database contains Zeus and all of his descendants, with $p$ now representing every child's probability of matching her parent's bit, a model that is commonly used for genetic correlations in evolutionary biology \citep{Cavender,Farris,Neyman}. 
\fi 

\paragraph{Our results.} 

Consider a policy-maker who 
\bkedit{specifies}
that an analyst must provide a certain differential privacy guarantee,
\bkedit{and wants to comprehend the \infpriv consequences of this policy for}
the population whose (correlated) data is being utilized. Our two main results can be interpreted as providing guidance to such a policy maker. The first result (\cref{t-pa}) supplies a closed-form expression for the inferential privacy guarantee as a function of the differential privacy parameter when data are {\em positively affiliated}\footnote{Positive affiliation (\cref{def:pa}) is a widely used notion of positive correlation amongst random variables. \bkedit{It is satisfied, for example, by graphical models whose edges encode positively-correlated conditional distributions on pairs of variables.}}\citep{MW82}.
The second result (\cref{t-bm}) allows understanding the behavior of the inferential privacy guarantee as a function of the degree of correlation in the population; it identifies a property of the joint distribution of data that ensures that the policy-maker can meet a given inferential privacy target via a differential privacy requirement that is a constant-factor scaling of that target. 

Among all mechanisms with a given differential privacy guarantee, which ones yield the worst inferential privacy when data are correlated? Our first main result, \cref{t-pa}, answers this question when data are positively affiliated, in addition to giving a closed-form expression for the inferential privacy guarantee. The answer takes the following form: we identify a simple property of mechanisms (\cref{def:max-biased}) such that any mechanism satisfying the property achieves the worst-case guarantee. Strikingly, the form of the worst-case mechanism does not depend on the joint distribution of the data, but {\em only} on the fact that the distribution satisfies positive affiliation.
We also provide one example of such a mechanism: a ``noisy-sum mechanism'' that simply adds Laplace noise to the sum of the data. 
This illustrates that the worst inferential privacy violations occur even with one of the most standard mechanisms for implementing differential privacy, rather than some contrived mechanisms.



The aforementioned results provide a sharp bound on the inferential privacy guarantee for positively affiliated distributions, but they say little about whether this bound is large or small in comparison to the differential privacy guarantee. Our second main result fills this gap: it provides an upper bound on the \infpriv guarantee when a {\em bounded affiliation} condition is satisfied on the correlations between individuals' rows in a database. Representing the strengths of these correlations by an {\em influence matrix} $\infmulmtx$, \cref{t-bm} asserts that if all row sums of this matrix are bounded by $1-\delta$ then every individual's \infpriv is bounded by $2 \idp/\delta$, regardless of whether or not the data are positively affiliated. Thus, \cref{t-bm} shows that in order to satisfy $\ndp$-\infpriv against all distributions with $(1-\delta)$-bounded affiliation, it suffices for the policy-maker to set $\idp = \delta \ndp / 2$. We complement this result with an example showing that the ratio of \infpriv to differential privacy can indeed be as large as $\Omega(\frac{1}{\delta})$, as the row sums of the influence matrix approach 1. Thus, the equivalence between inferential and differential privacy, $\ndp=\idp$, which holds for independent distributions, degrades gracefully to $\ndp = O(\idp)$ as one introduces correlation into the distribution, but only up to a point: as the row sums of the influence matrix approach 1, the ratio $\ndp/\idp$ can diverge to infinity, becoming unbounded when the row sums exceed 1. 

\paragraph{Techniques.} Our work exposes a formal connection between the analysis of inferential privacy in networked contexts and the analysis of spin systems in mathematical physics. In brief, application of a differentially private mechanism to correlated data is analogous to application of an external field to a spin system. Via this analogy, physical phenomena such as phase transitions can be seen to have consequences for data privacy: they imply that small variations in the amount of correlation between individuals' data, or in the differential privacy parameter of a mechanism, can sometimes have gigantic consequences for inferential privacy (\S\ref{sec:physics} elaborates on this point). Statistical physics also supplies the blueprint for \cref{t-bm} and its proof: our bounded affiliation condition can be regarded as a multiplicative analogue of Dobrushin's Uniqueness Condition \citep{Dob68,Dob70}, and our proof of \cref{t-bm} adapts the proof technique of the Dobrushin Comparison Theorem \citep{Dob70,Follmer82,Kunsch82} from the case of additive approximation to multiplicative approximation. \bkedit{Since Dobrushin's Uniqueness Condition is known to be one of the most
general conditions ensuring exponential decay of correlations 
in physics, our \autoref{t-bm} can informally be  interpreted as 
saying that differential privacy implies strong \infpriv guarantees when 
the structure of networked correlations is such that, conditional on the 
adversary's side information, the correlations between 
individuals' data decay rapidly as their distance in the network increases.}

\paragraph{Related work.}
\label{s-relwork}
\bkedit{Our paper adopts the term {\em \infpriv} as a convenient shorthand for a notion
that occurs in many prior works, dating back to \citet{Dal77}, which is
elsewhere sometimes called ``before/after privacy''~\cite{DPbook},
``semantic privacy''~\cite{KS14}, or ``noiseless privacy''~\cite{noiseless}. 
Dwork and McSherry observed that differentially private mechanisms supply \infpriv against 
adversaries whose prior is a product distribution; this was stated implicitly in~\cite{Dwork06} 
and formalized in~\cite{KS14}.
However, when adversaries can have arbitrary
auxiliary information, \infpriv becomes unattainable except by mechanisms that provide
little or no utility; see \cite{DworkNaor,nfl} for precise impossibility results along these 
lines. Responses to this predicament have varied: some works propose stricter notions
of privacy based on simulation-based semantics, \eg~{\em zero-knowledge privacy}~\cite{GLP_ZK}, 
others propose weaker notions based on restricting the set of prior distributions
that the adversary may have, \eg~{\em noiseless privacy}~\cite{noiseless}, 
and others incorporate aspects
of both responses, \eg~{\em coupled-world privacy}~\cite{bassily} and
the {\em Pufferfish} framework~\cite{pufferfish}. Our work is similar to some of the 
aforementioned ones in that we incorporate restrictions on the adversary's prior
distribution, however our goal is quite different:
rather than proposing a new privacy definition or a new class of
mechanisms, we quantify how effectively an existing class of mechanisms
($\eps$-differentially private mechanisms) achieves an existing privacy
goal (\infpriv).}

Relations between differential privacy and network analysis
have been studied by many authors---\eg \citep{KNRS13} and
the references therein---but this addresses a very different
way in which networks relate to privacy: the network in those
works is part of the data, whereas in ours it is a description of 
the auxiliary information.

The exponential mechanism of McSherry and Talwar~\cite{mt07}
can be interpreted in terms of Gibbs measures, and 
Huang and Kannan~\cite{HK12} leveraged this interpretation
and applied a non-trivial fact about free-energy minimization
to deduce consequences about incentive compatibility of
exponential mechanisms. Aside from their work, we are not
aware of other applications of statistical mechanics in
differential privacy.

\section{Defining Inferential Privacy}
\label{sec:def-ndp}

In this section we specify our notation and basic assumptions and 
definitions. A population of $n$ individuals is indexed
by the set $[n]=\{1,\ldots,n\}$. Individual $i$'s private 
data is represented by the element $\bg_i \in \vals$, where $\vals$ 
is a finite set. Except in \S\ref{sec:bounded} we will assume throughout, for simplicity,
that $\vals = \{0,1\}$, \ie each individual's private data is a 
single bit. When focusing on the networked 
privacy guarantee for a particular individual, we denote her index by $a \in [n]$
and sometimes refer to her as ``Athena''. 

A database is an $n$-tuple $\fdb \in \vals^n$ representing the private data
of each individual. As explained in \cref{sec:intro}, our model encodes
the `network' structure of the data using a probability distribution on
$\vals^n$; we denote this distribution by $\jd$. 
A computation performed on the database $\fdb$, whose outcome will be
disclosed to one or more parties, is called a {\em mechanism} and denoted
by $\M$. The set of possible outcomes of the computation
is $\outcomes$, and a generic outcome will be denoted by $\outc \in
\outcomes$. 
\%

\paragraph{Differential privacy~\cite{Dwork06,DMNS06,DPbook}.}
For a database $\fdb = (\bg_1,\ldots,\bg_n)$ 
and an individual $i \in [n]$, we use
$\fdb_{-i}$ to denote the $(n-1)$-tuple 
formed by omitting $\dbi$ from $\fdb$,
\ie $\fdb_{-i} = (\bg_1,\ldots,\bg_{i-1},\bg_{i+1},\ldots,\bg_n)$.
We define an equivalence relation $\sim_i$ by 
specifying that $\fdb \sim_i \adb \, \Leftrightarrow \,
\fdb_{-i} = \adb_{-i}$. For a mechanism $\M$ and individual $i$,
the differential privacy parameter $\idpi$ is defined by
\[
  e^\idpi = \max \left\{ \left. \frac{\Pr(\M(\fdb)=\outc)}{\Pr(\M(\adb) = \outc)} 
    \, \right| \,
    \fdb \sim_i \adb, \, \outc \in \outcomes \right\}.
\]
For any vector $\bm{\idp} = (\idp_1,\ldots,\idp_n)$ we say that 
$\M$ is $\bm{\idp}$-differentially private if 
the differential privacy parameter of $\M$ with respect to $i$ is at 
most $\idpi$, for every individual $i$. 

\paragraph{Inferential privacy.}
We define {\em \infpriv}
as an upper bound on the (multiplicative) change in 
$\frac{\Pr(\ba=\val_1)}{\Pr(\ba=\val_0)}$ when performing a Bayesian update
from the prior distribution $\jd$ to the posterior distribution
after observing $\M(\fdb) = \outc$. (If $\M$ has uncountably many potential
outcomes, we must instead consider doing a Bayesian update after observing
a positive-probability event $\M(\fdb) \in S$ for some set of outcomes $S$.) 
\begin{definition}
\label{def:ndp}
We say that mechanism $\M$ satisfies $\ndp$-\infpriv (with respect to individual $a$)
if the inequality $\frac{\Pr(\ba=\val_1 \given \M(\fdb) \in S)}{\Pr(\ba=\val_0 \given \M(\fdb) \in S)}
\leq e^{\ndp} \cdot \frac{\Pr(\ba=\val_1)}{\Pr(\ba=\val_0)}$
holds for all $\val_0,\val_1 \in X$ and all $S \subset \outcomes$ such that $\Pr(\M(\fdb) \in S)>0$. 
The {\em \infpriv parameter} of $\M$ is the smallest $\ndp$ with this property.
\end{definition}
\paragraph{Inferential versus differential privacy.}
A short calculation using Bayes' Law illuminates the relation between
these two privacy notions.
\begin{align}
\frac{\Pr(\ba=\val_1 \given \M(\fdb) \in S)}
     {\Pr(\ba=\val_0 \given \M(\fdb) \in S)} 
&=
\frac{\Pr(\M(\fdb) \in S \given \ba=\val_1)}
     {\Pr(\M(\fdb) \in S \given \ba=\val_0)}
\cdot
\frac{\Pr(\ba=\val_1)}{\Pr(\ba=\val_0)}.
\nonumber
\end{align}
Thus, the \infpriv parameter of mechanism $\M$ with respect
to individual $a$ is determined by:
\begin{equation} \label{eq:ndp}
e^{\ndp_a} = \sup \left\{ \left.
  \frac{\Pr(\M(\fdb) \in S \given \ba=\val_1)}
     {\Pr(\M(\fdb) \in S \given \ba=\val_0)}
  \, \right| \,
  \val_0,\val_1 \in X, \, \Pr(\M(\fdb) \in S) > 0 \right\}.
\end{equation}
Equivalently, if $\cdl,\cdu$ denote the conditional distributions
of $\fdb_{-a}$ given that $\ba=\val_0$ and $\ba=\val_1$, respectively,
then $\M$ is $\ndp_a$-inferentially private if 
\begin{align} \label{eq:def-ndp.1}
  \Pr(\M(\val_1,\pdb_1) \in S) &\leq e^{\ndp_a} \Pr(\M(\val_0,\pdb_0) \in S) \quad
  \mbox{when } \pdb_0 \sim \cdl, \, \pdb_1 \sim \cdu.\\
\intertext{For comparison, differential privacy asserts }
\label{eq:def-ndp.2}
  \Pr(\M(\val_1,\pdb) \in S) & \leq e^{\idpa} \Pr(\M(\val_0,\pdb) \in S) 
  \quad \forall \pdb.
\end{align}
When individuals' rows in the database are independent, 
$\cdl = \cdu$ and~\eqref{eq:def-ndp.2} implies~\eqref{eq:def-ndp.1}
with $\ndp_a = \idpa$ 
by averaging over $\pdb$. In other words, when bits are independent, 
$\idpa$-differential
privacy implies $\idpa$-inferential privacy. 
When bits are correlated, however, this implication breaks
down because the databases $\pdb_0,\pdb_1$ 
in~\eqref{eq:def-ndp.1} are sampled from different distributions. 
The `twins example' from \S\ref{sec:intro} illustrates
concretely why this makes a difference: if 
$\cdl$ and $\cdu$ are point-masses
on $(0,\ldots,0)$ and $(1,\ldots,1)$, respectively, 
then the \infpriv parameter of $\M$ is determined by the equation 
$e^{\ndp} = \sup_{S} 
\left\{ \frac{\Pr(\M(1,\ldots,1) \in S)}{\Pr(\M(0,\ldots,0) \in S)} \right\}$.
For an $\idp$-differentially-private mechanism this ratio may be as large 
as $e^{n \idp}$ since
the Hamming distance between $(0,\ldots,0)$ and $(1,\ldots,1)$ is $n$.

\fi 
\section{Positively Affiliated Distributions}
\label{sec:positively}
\label{s-pa}
Suppose a designer wants to ensure that Athena 
receives an \infpriv guarantee of 
$\ndp$, given a joint distribution $\jd$ on the 
data of individuals in the database. What is the 
largest differential privacy parameter $\idp$ that 
ensures this guarantee? The question is very
challenging even in the special case of
binary data (\ie, when $X=\{0,1\}$) because the ratio 
defining \infpriv (Equation~\ref{eq:def-ndp.1})
involves summing exponentially many terms in the
numerator and denominator. Determining the worst-case
value of this ratio over all differentially private
mechanisms $\M$ can be shown to be equivalent to solving a linear
program with exponentially many variables (the
probability of the event $\M(\fdb) \in S$ for every 
potential database $\fdb$) and exponentially many
constraints (a differential privacy constraint
for every pair of adjacent databases).

Our main result in this section answers this question when
individuals' data are 
binary-valued and {\em positively affiliated} \citep{FKG,MW82},
a widely used notion of positive correlation:
\cref{t-pa} gives a closed-form formula 
(Equation~\ref{eq:pa}) that one can invert to solve 
for the maximum differential privacy parameter $\idp$ 
that guarantees inferential privacy $\ndp$ when 
data are positively affiliated. The theorem
also characterizes the `extremal' mechanisms achieving the worst-case
\infpriv guarantee in~\eqref{eq:pa} as those satisfying a 
`maximally biased' property (\cref{def:max-biased}).
Intuitively, if one wanted to signal as strongly as possible 
that Athena's bit is 1 (resp., 0), a natural strategy---given 
that Athena's bit correlates positively with everyone else's---is 
to have a distinguished outcome (or set of outcomes) whose 
probability of being output by the mechanism increases with the 
number of 1's (resp., the number of 0's) in the database `as 
rapidly as possible', subject to differential privacy constraints. 
\cref{t-pa} establishes that this intuition is valid under 
the positive affiliation assumption. 
(Interestingly, the intuition is not valid if one merely
assumes that Athena's own bit is positively correlated
with every other individual's bit; see \cref{rmk:pos-corr}.)
\cref{lem:noisy-sum} provides one simple example
of a maximally-biased mechanism, namely a ``noisy-sum
mechanism'' that simply adds Laplace noise to the sum
of the bits in the database. Thus, the
worst-case guarantee in \cref{t-pa} is achieved 
not by contrived worst-case mechanisms, but by
one of the most standard mechanisms in the differential
privacy literature. 

We begin by defining positive affiliation, 
a concept that has proven extremely valuable in auction theory 
(the analysis of interdependent value auctions), 
statistical mechanics, and probabilistic combinatorics. 
Affiliation is a strong form of positive correlation between 
random variables: informally, positive affiliation means that 
if some individuals' bits are equal to 1 (or more generally, 
if their data is `large'), other individuals' bits are more 
likely to equal 1 as well (and similarly for 0).
We formally define positive affiliation for our setting below
and then state a key lemma concerning positively affiliated
distributions, the FKG inequality.

\begin{definition}[Positive affiliation] Given any two strings $\fdb_1, \fdb_2 \in \bn$, let $\fdb_1 \vee \fdb_2$ and $\fdb_1 \wedge \fdb_2$ 
denote their pointwise maximum and minimum, respectively. A  joint distribution $\jd$ on $\bn$ satisfies positive affiliation if  
\[ 
\jd(\fdb_1 \vee \fdb_2)\cdot\jd(\fdb_1 \wedge \fdb_2) ~~\geq~~ \jd(\fdb_1)\cdot\jd(\fdb_2)
\] 
for all possible pairs of strings $\fdb_1, \fdb_2$. Equivalently, $\jd$ satisfies positive affiliation if 
$\log \jd(\fdb) $ is a supermodular function of $\fdb \in \bn$.
\label{def:pa}
\end{definition} 


\begin{lemma}[FKG inequality; \citet{FKG}] \label{lem:fkg}
If $f,g,h$ are three real-valued functions on $\bn$ such that
$f$ and $g$ are monotone and $\log h$ is supermodular, then
\begin{equation} \label{eq:fkg}
  \left[ \sum_{\fdb} f(\fdb) g(\fdb) h(\fdb) \right] \,
  \left[ \sum_{\fdb} h(\fdb) \right] \,
  \geq \,
  \left[ \sum_{\fdb} f(\fdb) h(\fdb) \right] \,
  \left[ \sum_{\fdb} g(\fdb) h(\fdb) \right].
\end{equation}
\end{lemma}


In order to state the main result of this section, \cref{t-pa},
we must define a property that characterizes the mechanisms whose
\infpriv parameter meets the worst-case bound stated in the theorem. 
We defer the task of describing a mechanism that satisfies the
definition (or even proving that such a mechanism exists) until
\cref{lem:noisy-sum} below.
\begin{definition} \label{def:max-biased}
For $z \in \{0,1\}$, a mechanism $\M$ mapping $\bn$ to outcome set $\outcomes$ is 
called {\em maximally $z$-biased}, with respect to a vector of differential privacy parameters
$\bm{\idp} = (\idp_1,\ldots,\idp_n)$, if there exists a set of outcomes $S \subset \outcomes$ 
such that $\Pr(\M(\fdb) \in S) \propto \prod_{i=1}^n e^{- \idpi |\dbi-z|}$
for all $\fdb \in \bn$. In this case, we call $S$ a {\em distinguished outcome set} for $\M$.
\end{definition}

%
%

\begin{theorem} 
\label{t-pa}
Suppose the joint distribution $\jd$ satisfies positive affiliation. Then
for any $z \in \{0,1\}$ and any vector of differential privacy parameters,
$\bm{\idp} = (\idp_1,\ldots,\idp_n)$, the maximum of the ratio 
{\small \begin{equation} 
\label{eq:t-pa.ratio}
  \frac{\Pr(\M(\fdb) \in S \given \ba=z)}
       {\Pr(\M(\fdb) \in S \given \ba \neq z)},
\end{equation}}%
over all $\bm{\idp}$-differentially private mechanisms $\M$ and
outcome sets $S$,
is attained when $\M$ is maximally \mbox{$z$-biased}, with
distinguished outcome set $S$. Therefore,
the \infpriv guarantee to individual $a$ in the presence of 
correlation structure $\jd$
and differential privacy parameters $\idp_1,\ldots,\idp_n$
is given by the formula
{\small
\begin{equation} \label{eq:pa} 
\ndp_a = \max_{z \in \{0,1\}} \left\{
                     \ln \left|
                           \frac{\sum_{\fdb = (z,\pdb)} \jd^{z}({\pdb}) \exp \left( -\sum_{i=1}^n \idpi |\dbi-z| \right)}
                                  {\sum_{\fdb = (1-z,\pdb)} \jd^{1-z}({\pdb}) \exp \left( -\sum_{i=1}^n \idpi |\dbi-z| \right)}
                          \right| \right\}.
\end{equation}}%
\end{theorem} 
\begin{proof}
Suppose $z=0$ and consider any $\bm{\idp}$-differentially private mechanism $\M$
and outcome set $S$. Letting $\m(\fdb) = \Pr(\M(\fdb) \in S)$, we have the identity
{\small
\begin{equation} \label{eq:t-pa.0}
   \frac{\Pr(\M(\fdb) \in S \given \ba=0)}{\Pr(\M(\fdb) \in S \given \ba=1)} =
   \left. \frac{\Pr(\M(\fdb) \in S \mbox{ and } \ba=0)}{\Pr(\M(\fdb) \in S \mbox{ and } \ba=1)} 
   \right/  \frac{\Pr(\ba=0)}{\Pr(\ba=1)} =
   \frac{\sum_{\fdb = (0,\pdb)} \cdl(\pdb) \m(\fdb)}{\sum_{\fdb=(1,\pdb)} \cdu(\pdb) \m(\fdb )}.
\end{equation} }
When $\M$ is maximially 0-biased, with distinguished outcome set $S$,  
the right side of~\eqref{eq:t-pa.0}
is equal to $\frac{\sum_{\fdb=(0,\pdb)} \cdl(\pdb) e^{- \idpv \cdot \fdb}}{\sum_{\fdb=(1,\pdb)} \cdu(\pdb) e^{-\idpv \cdot \fdb}}$.
Thus, the $z=0$ case of the theorem is equivalent to the assertion that 
{\small \begin{equation} \label{eq:t-pa.2}
      \frac{\sum_{\fdb = (0,\pdb)} \cdl(\pdb) \m(\fdb)}{\sum_{\fdb=(1,\pdb)} \cdu(\pdb) \m(\fdb )} \leq
     \frac{\sum_{\fdb=(0,\pdb)} \cdl(\pdb) e^{- \idpv \cdot \fdb}}{\sum_{\fdb=(1,\pdb)} \cdu(\pdb) e^{-\idpv \cdot \fdb}}.
\end{equation}}%
After cross-multiplying and simplifying, this becomes
{\small \begin{equation} \label{eq:t-pa.3}
     \left[ \sum_{\fdb = (0,\pdb)} \jd(\fdb) \m(\fdb) \right] \, \cdot \,
     \left[ \sum_{\fdb=(1,\pdb)} \jd(\fdb) e^{-\idpv \cdot \fdb} \right] \; \leq \;
     \left[ \sum_{\fdb = (1,\pdb)} \jd(\fdb) \m(\fdb) \right] \, \cdot \,
     \left[ \sum_{\fdb=(0,\pdb)} \jd(\fdb) e^{-\idpv \cdot \fdb} \right].
\end{equation}}%
If we add $\left[ \sum_{\fdb = (1,\pdb)} \jd(\fdb) \m(\fdb) \right] \, \cdot \,
\left[ \sum_{\fdb=(1,\pdb)} \jd(\fdb) e^{-\idpv \cdot \fdb} \right]$ to
both sides, we find that~\eqref{eq:t-pa.3} is equivalent to
{\small \begin{equation} \label{eq:t-pa.4}
     \left[ \sum_{\fdb \in \bn} \jd(\fdb) \m(\fdb) \right] \, \cdot \,
     \left[ \sum_{\fdb=(1,\pdb)} \jd(\fdb) e^{-\idpv \cdot \fdb} \right] \; \leq \;
     \left[ \sum_{\fdb = (1,\pdb)} \jd(\fdb) \m(\fdb) \right] \, \cdot \,
     \left[ \sum_{\fdb \in \bn} \jd(\fdb) e^{-\idpv \cdot \fdb} \right].
\end{equation}}%
To prove~\eqref{eq:t-pa.4} we will apply the FKG inequality. 
Set $h(\fdb) = \jd(x) e^{-\idpv \cdot \fdb}$ and note that
$\log h$ is the sum of $\log \jd$---a supermodular function---and
$(-\idpv) \cdot \fdb$, a linear function. Hence $\log h$ is 
supermodular. Now define $f(\fdb) =  \m(\fdb) e^{\idpv \cdot \fdb}$ and 
$g(\fdb) = \dba$. The differential
privacy constraint for $\m$ implies that 
$f$ is monotonically non-decreasing; observe that
$g$ is monotonically non-decreasing as well.
The FKG inequality implies
\begin{equation} \label{eq:t-pa.5}
  \left[ \sum_\fdb f(\fdb) h(\fdb) \right] \,
  \left[ \sum_\fdb g(\fdb) h(\fdb) \right] \leq
  \left[ \sum_\fdb f(\fdb) g(\fdb) h(\fdb) \right] \,
  \left[ \sum_\fdb h(\fdb) \right].
\end{equation}
Substituting the definitions of $f, \, g, \, h$ into~\eqref{eq:t-pa.5}
we readily see that it is equivalent to~\eqref{eq:t-pa.4}, which
completes the proof.
\end{proof}

\noindent
Finally, as promised at the start of this section, we show
that a noisy-sum mechanism that adds Laplace noise to the 
sum of the bits in the database is maximally $z$-biased
for every $z \in \{0,1\}$. Together with \cref{t-pa},
this shows that any \infpriv guarantee that can be proven
for the noisy-sum mechanism automatically extends to
a guarantee for all differentially private mechanisms,
when data are positively affiliated.
\begin{lemma} \label{lem:noisy-sum}
Suppose that all individuals have the same differential privacy parameter, \ie that 
$\bm{\idp} = (\idp,\idp,\ldots,\idp)$ for some $\idp>0$. Consider the {\em noisy-sum
mechanism} $\NS$ that samples a random $Y$ from the Laplace distribution with
scale parameter $1/\idp$ and outputs the sum $Y+\sum_{i=1}^n \dbi$. For all
$z \in \{0,1\}$ the mechanism $\NS$ is maximally $z$-biased.
\end{lemma}
\begin{proof}
For any $\fdb \!\in\!  \bn$, let $|\fdb|  \! = \! \sum_{i=1}^n \dbi$. 
When $z\!=\!0$ and $\bm{\idp} \!=\! (\idp,\idp,\ldots,\idp)$, 
the definition of a maximally $z$-biased mechanism 
requires the existence of an outcome set $S$ such
that $\Pr(\NS(\fdb) \in S) \propto e^{-\idp |\fdb|}$.
For the set $S = (-\infty,0]$, the event $\NS(\fdb) \in S$ coincides with the event 
$Y \leq -|x|$. Since $Y$ is a Laplace random variable with scale parameter $1/\idp$,
this event has probability proportional to $e^{-\idp |x|}$, as desired.
When $z=1$ the proof of the lemma proceeds identically, using the set 
$S = [n,\infty)$.
\end{proof}

\begin{remark} \label{rmk:pos-corr}
Intuitively, one might expect \cref{t-pa} to hold
whenever the joint distribution $\jd$ is such that each
pair of bits is positively correlated, a weaker property
than positive affiliation which requires each pair of bits
to be positively correlated {\em even after conditioning on
any possible tuple of values for the remaining bits}. In
Appendix~\ref{app:pairwise} we present an example 
illustrating that the theorem's conclusion can be violated
(in fact, quite drastically violated) when one only assumes
pairwise positive correlation. The basic reason
is that when bits are pairwise positively correlated,
it may still be the case that one individual's bit correlates
much more strongly with a non-monotone function of the others'
bits than with any monotone function.
\end{remark}

\begin{remark} \label{rmk:ising}
The quantities appearing in \cref{t-pa} have precise analogues
in the physics of spin systems, and this analogy sheds light on
\infpriv. Appendix~\ref{sec:physics} delves into this connection in
detail; in this remark we merely sketch a dictionary for translating
between \infpriv and statistical mechanics and discuss some
consequences of this translation.

In brief, an adversary's prior distribution on $\bn$ corresponds 
to the Gibbs measure of a two-spin system with
Hamiltonian $H(\fdb) = - \ln \jd(\fdb)$. Under this
correspondence, positively affiliated distributions
correspond to ferromagnetic spin systems. The adversary's
posterior distribution after applying a maximally
0-biased (resp., maximally 1-biased) mechanism is
equivalent to the Gibbs measure of the spin system
after applying the external field $\frac12 \idpv$
(resp., $-\frac12 \idpv$). The worst-case \infpriv
guarantee for Athena in \cref{t-pa} is therefore equivalent 
(up to a bijective transformation) to the magnetization 
at Athena's site when the external field $\pm \idpv$
is applied to the spin system.

One of the interesting implications of this correspondence
concerns phase transitions. Statistical-mechanical systems such as
magnets are known to undergo sharp transitions in their
physical properties as one varies thermodynamic quantities
such as temperature and external field strength. Translating
these results from physics to the world of privacy using the
dictionary outlined above, one discovers that
\infpriv guarantees can undergo surprisingly sharp variations
as one varies a mechanism's differential privacy parameter or
an adversary's belief about the strength of correlations between 
individuals' bits in a database. \cref{t-ising} in the appendix
formalizes these observations about phase transitions in \infpriv.
\end{remark}

\ify
\section{Bounded Affiliation Distributions}
\label{sec:bounded}

In this section we present a general upper bound for 
\infpriv
that applies under a condition that we call {\em bounded affiliation}.
Roughly speaking, bounded affiliation requires that correlations
between individuals are sufficiently weak, in the sense that
the {\em combined influence} of all other individuals on any 
particular one is sufficiently small. A very similar 
criterion in the statistical mechanics literature,
{\em Dobrushin's uniqueness condition} \citep{Dob68,Dob70},
is identical to ours except that it defines 
``influence'' in terms of additive approximation 
and we define it multiplicatively (\autoref{def:inf}).
\citeauthor{Dob70} showed that this condition implies
uniqueness of the Gibbs measure for a specified collection of
conditional distributions. Its implications for correlation
decay \citep{Gross79,Follmer82,Kunsch82} and mixing times of Markov
chains \citep{AizenmanHolley,Weitz05,Hayes06} were subsequently
explored.
Indeed, our proof of network differential 
privacy under the assumption of bounded affiliation
draws heavily upon the methods of \citet{Dob70}, \citet{Gross79}, 
and \citet{Kunsch82} on decay of correlations under
Dobrushin's uniqueness condition.

Throughout this section (and its corresponding appendix) we assume
that each individual's private data belongs to a finite set $X$
rather than restricting to $X = \{0,1\}$. This assumption does
not add any complication to the theorem statements and proofs,
while giving our results much greater generality. 
We now define the notion of influence that is relevant to
our results on distributions with bounded affiliation.
\begin{definition} \label{def:inf}
If $\bg_0,\ldots,\bg_n$ are jointly distributed random variables,
the {\em multiplicative influence} of $\bg_j$ on $\bg_i$, denoted
by $\infmul_{ij}$, is defined by the equation
{\small \[
    e^{2 \infmul_{ij}} = \max \left\{ \left.
                    \frac{\Pr(\bg_i \in S \mid \fdb_{-i})}
                         {\Pr(\bg_i \in S \mid \adb_{-i})}
                    \; \right| \;
                    S \subseteq \supp(\bg_i), \,
                    \fdb_{-i} \sim_j \adb_{-i}
                    \right\}.
\] }
In other words, the influence of $\bg_j$ on $\bg_i$ is one-half of the 
(individual) differential privacy parameter of $\bg_i$ with
respect to $\bg_j$, when one regards $\bg_i$ as a randomized
function of the database $\fdb_{-i}$. When $i=j$ one adopts
the convention that $\infmul_{ij}=0$.
The {\em multiplicative influence matrix} 
is the matrix
$\infmulmtx = (\infmul_{ij})$.
\end{definition}

\begin{theorem} \label{t-bm}
Suppose that the joint distribution $\jd$ has a multiplicative
influence matrix $\infmulmtx$ whose spectral norm is strictly
less than 1. Let $\invmulmtx = (\invmulentry_{ij})$ denote the
matrix inverse of $I-\infmulmtx$. Then for any mechanism with
individual privacy 
parameters $\bm{\idp} = (\idpi)$, 
the \infpriv guarantee
satisfies
\ifextabs
\begin{equation} \label{eq:ndp-dobr-body}
\forall i \;\; 
\ndp_i \leq 2 \sum_{j=1}^{n} \invmulentry_{ij} \idp_j.
\end{equation}
\else
\begin{equation} \label{eq:ndp-dobr-body}
\forall i \;\; 
\ndp_i \leq 2 \sum_{j=1}^{n} \invmulentry_{ij} \idp_j.
\end{equation}
\fi
If the matrix of multiplicative influences 
satisfies 
\ifextabs
$
\forall i \;
\sum_{j=1}^{n} \infmul_{ij} \idp_j \leq (1-\delta) \idpi
$
\else
\begin{equation} \label{eq:ba-cond-body}
\forall i \;\;
\sum_{j=1}^{n} \infmul_{ij} \idp_j \leq (1-\delta) \idpi
\end{equation}
\fi
\xspace for some $\delta>0$, then 
\ifextabs
$\ndp_i \leq 2 \idpi / \delta$ for all $i$.
\else
\begin{equation} \label{eq:ndp-ba-body}
\forall i \;\; \ndp_i \leq 2 \idpi / \delta.
\end{equation}
\fi
\end{theorem}

\begin{proof}[Proof sketch.]
Let $S$ be any set of
potential outcomes of the mechanism $\M$ such that
 $\Pr(\M(\fdb) \in S)>0$. Let $\cd{1}$ denote the 
conditional distribution on databases $\fdb \in X^n$,
given that $\M(\fdb) \in S$, and let
$\cd{2}$ denote the unconditional distribution $\jd$,
respectively.
For $i \in \{1,2\}$ and for any function $f : X^n
\to \reals$, let $\cd{i}(f)$
denote the expected value of $f$ under distribution
$\cd{i}$. Also define the Lipschitz constants
$
\lip_i(f) = \max 
            \{ f(\fdb) - f(\adj{\fdb}) 
               \; \mid \; \fdb \sim_i \adj{\fdb} \}.
$
\ifextabs
The heart of the proof lies in showing that if
$f$ takes values in $\reals_+$ then
\begin{equation} \label{eq:fest-body}
|\ln \cd{1}(f) - \ln \cd{2}(f)| \leq \frac12
\sum_{i,j=1}^n \invmtx_{ij} \idp_j \lip_i(\ln f).
\end{equation}
This is done by studying the set of all vectors
$\bm{\estim}$ that satisfy 
$|\ln \cd{1}(f) - \ln \cd{2}(f)| \leq 
 \sum_{i=1}^n \estim_i \lip_i(f)$
for all $f$, and showing that this
set is non-empty and is preserved by an affine transformation $T$
that is a contracting mapping of $\reals^n$ 
(when the spectral norm of $\infmulmtx$ is less than 1)
with fixed point $\frac12 \invmtx \bm{\idp}$.
To derive~\eqref{eq:ndp-dobr-body} 
from~\eqref{eq:fest-body}, use the definition
of $\ndp_i$ to choose two
distinct values $\val_0 \neq \val_1$ in $\vals$
such that $\ndp_i =
\left| \ln \left( \frac{
         \Pr(\dbi=\val_1 \mid \M(\fdb) \in S) \, / \, 
         \Pr(\dbi=\val_0 \mid \M(\fdb) \in S) 
       }{
         \Pr(\dbi=\val_1) \, / \, \Pr(\dbi=\val_0)
       }        
       \right)
\right| = 
\left| \ln \left( 
  \frac{ \cd{1}(f) / \cd{1}(g) }{ \cd{2}(f) / \cd{2}(g) }
\right) \right|,$ where $f,g$ are the indicator functions
of $\dbi=\val_0$ and $\dbi=\val_1$, respectively. Unfortunately
$\lip_i(\ln f) = \lip_i(\ln g) = \infty$ so direct application 
of~\eqref{eq:fest-body} is not useful; instead, we define a suitable
averaging operator $\avg$ to smooth out $f$ and $g$, thereby 
improving their Lipschitz constants and enabling application
of~\eqref{eq:fest-body}. A separate argument is then used to
bound the error introduced by smoothing $f$ and $g$ using $\avg$,
which completes the proof of~\eqref{eq:ndp-dobr-body}. 
Under the hypothesis that $\infmulmtx \bm{\idp} \preceq (1-\delta)
\bm{\idp}$, the relation $\bm{\ndp} \preceq \frac{2}{\delta} \bm{\idp}$
is easily derived from~\eqref{eq:ndp-dobr-body}
by applying the formula
$\invmtx = \sum_{m=0}^{\infty} \infmulmtx^m$.
The full proof is presented in Appendix~\ref{app:bounded}.
\else
Also define the Lipschitz constants
\begin{equation} \label{eq:lip}
\lip_i(f) = \max 
            \{ f(\fdb) - f(\adj{\fdb}) 
               \; \mid \; \fdb \sim_i \adj{\fdb} \}.
\end{equation}
Let us say that a vector 
$\bm{\estim} = (\estim_{i})$ is a {\em multiplicative estimate}
if for every function $f : X^n \to \reals_+$ we have
\begin{equation} \label{eq:est}
| \ln \cd{1}(f) - \ln \cd{2}(f) | \leq \sum_{i=1}^{n} \estim_i 
          \lip_i(\ln f).
\end{equation}
This definition is the multiplicative analogue of the notion of 
``estimate'' used by \citet{Dob70}, \citet{Follmer82}, and \citet{Kunsch82} 
in their proofs of the so-called Dobrushin Comparison Theorem.
Similar to those proofs, 
the heart of our proof will lie in showing that the affine transformation 
\begin{equation} \label{eq:T}
T(\bm{\estim}) = \tfrac1{2n} \bm{\idp} + 
  \left( 1 - \tfrac{1}{n} \right) \bm{\estim} + \tfrac1n \infmulmtx \bm{\estim}
\end{equation}
maps multiplicative estimates to multiplicative estimates. 
We defer the proof of this fact to \autoref{lem:T} 
in Appendix~\ref{app:bounded}. In the sequel let 
$\cmtx = \left( 1 - \tfrac{1}{n} \right)I + \tfrac1n \infmulmtx$,
so that $T(\bm{\estim}) = \tfrac{1}{2n} \bm{\idp} + \cmtx \bm{\estim}$.

The set of multiplicative estimates is non-empty.
For example, for any $f : X^n = \reals_+$ we have
\[
| \ln \cd{1}(f) - \ln \cd{2}(f) | \leq
| \sup(\ln f) - \inf(\ln f) | \leq
\sum_{i=1}^n \lip_i(\ln f),
\]
implying that $\estim_i=1$ for $i=1,\ldots,n$
defines a multiplicative estimate.
Iterating the mapping $T$ starting from an
arbitrary $\bm{\estim}$ we obtain the sequence
\[
\bm{\estim} \; \mapsto \;
\tfrac{1}{2n} \bm{\idp} + \cmtx \estim \; \mapsto \;
\tfrac{1}{2n} \bm{\idp} + \tfrac{1}{2n} \cmtx \bm{\idp} + \cmtx^2 \estim 
  \; \mapsto \;
\tfrac{1}{2n} \bm{\idp} + \tfrac{1}{2n} \cmtx \bm{\idp} + 
\tfrac{1}{2n} \cmtx^2 \bm{\idp} + \cmtx^3 \estim \; \mapsto \;
\cdots
\]
Under the assumption that $\infmulmtx$ has spectral norm strictly 
less than 1---and hence the convex
combination $\cmtx = \left( 1 - \tfrac{1}{n} \right) I + \tfrac{1}{n}
\infmulmtx$ also has spectral norm less than 1---this sequence 
converges to 
\[
  \tfrac{1}{2n} 
  \left( \sum_{n=0}^{\infty} \cmtx^n \right) \bm{\idp} 
  = \tfrac{1}{2n} (I - \Upsilon)^{-1} \bm{\idp} 
  = \tfrac{1}{2n} \left( \tfrac1n (I - \infmulmtx) \right)^{-1} \bm{\idp}
  = \tfrac12 \invmtx \bm{\idp},
\]
hence $\tfrac12 \invmtx \bm{\idp}$ is also a multiplicative estimate. In other
words, for any $f : \{0,1\}^{n+1} \to \reals_{+}$ it holds that
\begin{equation} \label{eq:fest}
\left| \ln \cd{1}(f) - \ln \cd{2}(f) \right| \leq 
\tfrac12 \sum_{i,j=1}^{n} \invmtx_{ij} \idp_j \lip_i(\ln f).
\end{equation}
To prove~\eqref{eq:ndp-dobr}, the most natural idea is to apply
this estimate to the function $f(\fdb) = \dbi$ \bkmargincomment{say why?}
but unfortunately $\lip_i(f) = \infty$, so we instead apply an 
averaging operator to smooth $f$ before invoking~\eqref{eq:fest};
the definition of this operator and the details of this calculation
are deferred to Appendix~\ref{app:bounded}.
\bkmargincomment{appendix!}
\bkcomment{This is the one place that I need the assumption
that $X = \{0,1\}$. In general we would want to replace $f$
and $1-f$ with functions that take the value 1 if and only if
$x_i$ takes the value $v_0$ or $v_1$, respectively. This 
requires a bit more notation and I am saving for later the
task of modifying this part.}
Noting that $\lip_k(f) = \lip_k(\omf) = 0$ for $k \neq i$, 
we see that~\eqref{eq:fest} implies
\[
\left| \ln \cd{1}(f) - \ln \cd{2}(f) \right| \leq
\tfrac12 \sum_{j=1}^n \invmtx_{ij} \idp_j 
  \qquad \text{and} \qquad
\left| \ln \cd{1}(\omf) - \ln \cd{2}(\omf) \right| \leq
\tfrac12 \sum_{j=1}^n \invmtx_{ij} \idp_j .
\]
Hence
\begin{align*}
\left| \ln \left( \frac{\Pr(\M(\fdb) \in S \mid \dbi=1)}
                {\Pr(\M(\fdb) \in S \mid \dbi=0)} \right) \right| 
 &=
\left| \ln \left( \frac{\Pr(\dbi=1 \mid \M(\fdb) \in S)}
                {\Pr(\dbi=0 \mid \M(\fdb) \in S)} 
           \, \cdot \,
           \frac{\Pr(\dbi=0)}{\Pr(\dbi=1)} 
\right) \right| \\
&=
\left| \ln \left( \frac{\cd{1}(f)}{\cd{1}(\omf)} \, \cdot \, 
                  \frac{\cd{2}(\omf)}{\cd{2}(f)} \right) \right| \\
&\leq
\left| \ln \cd{1}(f) - \ln \cd{2}(f) \right| + 
\left| \ln \cd{1}(\omf) - \ln \cd{2}(\omf) \right| \\
&\leq
\sum_{j=1}^n \invmtx_{ij} \idp_j
\end{align*}
which completes the proof of~\eqref{eq:ndp-dobr}.

To prove~\eqref{eq:ndp-ba}, we use the partial
ordering on vectors defined by $\bm{a} \preceq \bm{b}$ if 
and only if $a_i \leq b_i$ for all $i$. 
The matrix $\infmulmtx$ has non-negative entries,
so it preserves this ordering: if $\bm{a} \preceq \bm{b}$ then 
$
   \forall i \; \sum_{j} \infmul_{ij} a_j \leq \sum_{j} \infmul_{ij} b_j
$
and hence $\infmulmtx \bm{a} \preceq \infmulmtx{b}$. 
Rewriting the relation~\eqref{eq:ba-cond} in the form
 $\infmulmtx \bm{\idp} \preceq (1-\delta) \bm{\idp}$
and applying induction, we find that for all $n \geq 0$,
 $\infmulmtx^n \bm{\idp} \preceq (1-\delta)^n \bm{\idp}$.
Summing over $n$ yields
\[
 \invmtx \bm{\idp} = \sum_{n=0}^{\infty} \infmulmtx^n \bm{\idp} \preceq
 \sum_{n=0}^{\infty} (1-\delta)^n \bm{\idp} = \tfrac{1}{\delta} \bm{\idp}
\]
which, when combined with~\eqref{eq:ndp-dobr}, yields~\eqref{eq:ndp-ba}. 
\fi
\end{proof}
\noindent
The bound $\ndp_i \leq 2 \idpi / \delta$ in the theorem is tight
up to a constant factor. This is shown in \S\ref{sec:physics} by considering
an adversary whose prior is
the Ising model of a complete $d$-ary tree $T$ at inverse temperature 
$\beta = \tanh^{-1} \left( \frac{1-\delta}{d} \right)$
The entries of the influence matrix satisfy 
$\infmul_{ij} = \beta$ if $(i,j) \in E(T)$, 0 otherwise.
Thus, the row sum $\sum_{j=1}^{n} \infmul_{ij}$ is maximized
when $i$ is an internal node, with degree $d+1$, in which
case the row sum is $(d+1) \tanh^{-1} \left( \frac{1-\delta}{d} \right)
= 1 - \delta - o(1)$ as $d \to \infty$. 
In \S\ref{sec:physics} we apply \cref{t-pa} to show that the 
\infpriv guarantee for the Ising model
on a tree satisfies $\ndp = \Omega(\frac{\idp}{\delta})$, 
matching the upper bound in \autoref{t-bm} up to a constant factor.




\section{Conclusion} 

A number of immediate questions are prompted by our results, such as incorporating 
$(\eps,\delta)$-privacy into our analysis of inferential guarantees (for product 
distributions this was achieved in~\cite{KS14}) and
extending the analysis in \S\ref{sec:positively} to non-binary databases where an individual's data cannot be summarized by a single bit. A key challenge here is to find an analogue of positive affiliation for databases whose rows cannot naturally be interpreted as elements of a lattice. More excitingly, however, the scenario of datasets with networked correlations raises several broad directions for future work. 

\noindent
\bkedit{
{\bf Designing for inferential privacy:} Our work takes differentially private algorithms as a primitive and {\em analyzes} what inferential privacy is achievable with given differential privacy guarantees. This allows leveraging the vast body of work on, and adoption of, differentially private algorithms, while remaining agnostic to the data analyst's objective or utility function. However if one instead assumes a particular measure of utility, one can directly
investigate the design of inferential-privacy preserving algorithms to obtain stronger guarantees:
given some joint distribution(s) and utility objectives, what is the best inferential privacy achievable,
and what algorithms achieve it? 
}

\noindent
{\bf Inferential privacy and network structure:} An intriguing set of questions arises from returning to the original network structures that led to the model of correlated joint distributions. Note that our results in \autoref{t-pa} give the inferential privacy guarantee for a particular individual: how do inferential privacy guarantees depend on the {\em position} of an individual in the network (for instance, imagine the central individual in a large star graph versus the leaf nodes), and how does the relation between the correlations and the network structure play in? 

\paragraph{Acknowledgements}
The authors gratefully acknowledge helpful discussions with Christian Borgs, danah boyd, 
Jennifer Chayes, Cynthia Dwork, Kobbi Nissim, Adam Smith, Omer Tamuz, and Salil Vadhan. 

The authors acknowledge the support of NSF awards AF-1512964 and III-1513692, 
and ONR award N00014-15-1-2335. Robert Kleinberg gratefully acknowledges
the support of Microsoft Research New England, where he was employed while most of this
research took place.

\fi 

\bibliography{ndp}

\appendix

\section{Appendix to \S\ref{sec:positively}: Positively Affiliated Distributions}
\label{app:positively}

This appendix contains material accompanying \S\ref{sec:positively}
that was omitted from that section for space reasons.

\subsection{Pairwise positive correlation}
\label{app:pairwise}

A weaker condition than positive affiliation is {\em pairwise
positive correlation}. This property of a joint distribution
$\jd$ on databases $\fdb \in \bn$ requires that for each pair 
of indices $i,j \in [n]$, the (unconditional) 
marginal distribution of the bits $\bg_i, \bg_j$ satisfies
$$
  \expect[\bg_i \bg_j] \geq \expect[\bg_i] \cdot \expect[\bg_j].
$$
If the inequality is strict for every $i,j$ then we say
$\jd$ is {\em pairwise strictly positively correlated.}

Recall Theorem~\ref{t-pa}, which establishes that when 
a joint distribution 
$\jd$ satisfies positive affiliation then the worst-case
\infpriv guarantee is attained by any maximally $z$-biased
distribution. The intuition supporting the theorem statement
might seem to suggest that the same conclusion holds 
whenever $\jd$ satisfies pairwise positive correlation.
In this section we show that this is not the case: 
if $\jd$ satisfies pairwise positive correlation 
(or even strict pairwise positive correlation) there
may be a mechanism whose \infpriv guarantee is much
worse than that of any maximally $z$-biased mechanism.

Our construction applies when $n$ is of the form 
$n = 1+rs$ for two positive integers $r,s$. For a
database $\fdb \in \bn$ we will denote one of its
entries by $\dba$ and the others by $\bg_{ij}$ 
for $(i,j) \in [r] \times [s]$. The joint 
distribution $\jd$ is uniform over the solution
set of the system of congruences
\begin{equation} \label{eq:ppc.1}
   \dba + \sum_{j=1}^s \bg_{ij} \equiv 0 \pmod{2} 
   \qquad \mbox{for } i=1,\ldots,r
\end{equation}
Thus, to sample from $\jd$ one 
draws the bits $\dba$ and $\bg_{ij}$ for $(i,j) \in [r] \times [s-1]$
independently from the uniform distribution
on $\{0,1\}$, then one sets $\bg_{is}$ for 
all $i$ so as to satisfy~\eqref{eq:ppc.1}.

The distribution $\jd$ is pairwise independent,
hence it is pairwise positively correlated.
(The calculation of privacy parameters is much
easier in the pairwise-independent case. 
At the end of this section we apply a simple
continuity argument to modify the
example to one with pairwise strict positive
correlation without significantly changing the
privacy parameters.)

Let us first calculate the \infpriv for a 
mechanism $\M_1$ that calculates the number
of odd integers in the sequence
\begin{equation} \label{eq:ppc.2}
\dba, \, \sum_{j=1}^s \bg_{1j}, \, \sum_{j=1}^s \bg_{2j}, \, 
\ldots, \sum_{j=1}^s \bg_{rj}
\end{equation}
 and adds
Laplace noise (with scale parameter 
$1/\idp$) to the result. This is $\idp$-differentially
private since changing a single bit of $\fdb$
changes the parity of only one element of the 
sequence. However, when $\fdb$ is sampled
from $\jd$ the number of odd integers in 
the sequence~\eqref{eq:ppc.2} is either 0
if $\dba=0$ or $n$ if $\dba=1$. Hence
\begin{align*}
\Pr(\M_1(\fdb) \leq 0 \given \dba=0) &= \tfrac12 \\
\Pr(\M_1(\fdb) \leq 0 \given \dba=1) &= \tfrac12 e^{-(r+1)\idp}
\end{align*}
implying that the \infpriv parameter of $\M_1$ is
at least $(r+1)\idp$.

Now let us calculate the \infpriv parameter of a maximally
0-biased mechanism $\M_2$, with outcome $\outc \in \outcomes$ 
such that
$\Pr(\M_2(\fdb) = \outc \given \fdb) \propto e^{-\idp |\fdb|}$,
where $|\fdb|$ denotes the sum of the bits in $\fdb$.
Let $T_0$ (resp.\ $T_1$) denote the set of bit-strings
in $\{0,1\}^s$ having even (resp.\ odd) sum, and let
$T_0^r, \, T_1^r$ denote the $r^{\mathrm{th}}$ Cartesian
powers of these sets. The conditional distribution
of $(\bg_{ij})$ given $\dba=0$ is the uniform distribution
on $T^r_0$, and the conditional distribution
of $(\bg_{ij})$ given $\dba=1$ is the uniform distribution
on $T^r_1$. For $\pdb = (y_{ij}) \in \{0,1\}^{rs}$ and $i \in [r]$,
let $\pdb_{i\ast}$ denote the 
$s$-tuple $(y_{i1},\ldots,y_{is})$. We have
\begin{align}
\nonumber
\Pr(\M_2(\fdb) = \outc \given \dba = 0) &= 
\sum_{\fdb = (0,\pdb)} \Pr(\M_2(\fdb) = \outc \given \fdb) \cdot 
                    \Pr(\fdb \given \dba=0) 
  \\ &=
\nonumber
\sum_{\pdb \in T_0^r} e^{-\idp |\pdb|} \cdot 2^{-r(s-1)} 
  \\ &=
\nonumber
\sum_{\pdb \in T_0^r} 
\prod_{i=1}^r \left( e^{-\idp |\pdb_{i\ast}|} \cdot 2^{1-s} \right)
  \\ &=
\label{eq:ppc.3}
\left( 2^{(1-s)} \sum_{\bm{z} \in T_0} e^{-\idp |\bm{z}|} \right)^r.
\end{align}
Similarly,
\begin{equation} \label{eq:ppc.4}
\Pr(\M_2(\fdb) = \outc \given \dba = 1) = 
  e^{-\idp} \cdot \left( 2^{(1-s)} \sum_{\bm{z} \in T_1} e^{-\idp |\bm{z}|} \right)^r.
\end{equation}
(The extra factor of $e^{-\idp}$ on the right side comes from
the fact that $\dba=1$, which inflates the exponent in the
expression $e^{-\idp |\fdb|}$ by $\idp$.)
To evaluate the expressions on the right sides 
of~\eqref{eq:ppc.3}-\eqref{eq:ppc.4}, it is useful
to let $A_0 = \sum_{\bm{z} \in T_0} e^{-\idp |\bm{z}|}$
and $A_1 = \sum_{\bm{z} \in T_1} e^{-\idp |\bm{z}|}$. Then
we find that 
\begin{align*}
A_0 + A_1 &= 
\sum_{\bm{z} \in \{0,1\}^s} e^{-\idp |\bm{z}|} =
(1 + e^{-\idp})^s \\
A_0 - A_1 &=
\sum_{\bm{z} \in \{0,1\}^s} (-1)^|\bm{z}| \cdot e^{-\idp |\bm{z}|} =
(1 - e^{-\idp})^s \\
A_0 &=
\frac12 \left[ (1 + e^{-\idp})^s + (1 - e^{-\idp})^s \right] \\
A_1 &= 
\frac12 \left[ (1 + e^{-\idp})^s - (1 - e^{-\idp})^s \right].
\end{align*}
Substituting these expressions into~\eqref{eq:ppc.3}-\eqref{eq:ppc.4}
we may conclude that
\begin{equation}
\frac{ \Pr(\M_2(\fdb) = \outc \given \dba = 0) }
     { \Pr(\M_2(\fdb) = \outc \given \dba = 1) } 
  =
\frac{ e^{\idp} \cdot A_0^r } { A_1^r }
  = 
e^{\idp} \left[ \frac{ (1 + e^{-\idp})^s + (1 - e^{-\idp})^s }
                    { (1 + e^{-\idp})^s - (1 - e^{-\idp})^s }
        \right]^r.
\end{equation}
The \infpriv parameter of $\M_2$ is therefore given by
\begin{align*}
\ndp &= \idp + r \ln \left[ \frac{ (1 + e^{-\idp})^s + (1 - e^{-\idp})^s }
                    { (1 + e^{-\idp})^s - (1 - e^{-\idp})^s }
        \right] 
  \\ & < 
\idp + r \ln \left[ 1 + 2 \left( \frac{1-e^{-\idp}}{1+e^{-\idp}} \right)^s \right]
  \\ & =
\idp + r \ln \left[ 1 + 2 \tanh^s(\idp) \right] <
\idp + 2 r \idp^s.
\end{align*}
Comparing the inferential privacy parameters of $\M_1$ and $\M_2$,
they are $(r+1)\idp$ and $\idp + 2r \idp^s$, respectively, so the
inferential privacy parameter of $\M_1$ exceeds that of the maximally
0-biased mechanism, $\M_2$, by an unbounded factor as $r,s \to \infty$.

Under the distribution $\jd$ we have analyzed thus far, the bits of 
$\fdb$ are pairwise independent. However, we may take a convex
combination of $\jd$ with any distribution in which all pairs of bits
are strictly positively correlated---for example, a distribution
that assigns equal probability to the two databases
$(1,\ldots,1)$ and $(0,\ldots,0)$ and zero probability 
to all others. In this way we obtain a distribution $\jd'$ which
satisfies pairwise strict positive correlation and may can be made
arbitrarily close to $\jd$ by varying the mixture parameter of the
convex combination. Since the \infpriv parameter of a mechanism 
with respect to a given prior distribution is
a continuous function of that distribution,
it follows that the \infpriv parameters of $\M_1$
and $\M_2$ can remain arbitrarily close to the values calculated
above while imposing a requirement that the prior on
$\fdb$ satisfies pairwise strict positive correlation.



\subsection{Connection to Ferromagnetic Spin Systems}
\label{sec:physics}

\label{sec:gibbs}


The quantities appearing in \cref{t-pa}
have precise analogues in the physics of spin
systems, and this analogy sheds light on
\infpriv. In statistical mechanics, a
two-spin system composed of $n$ sites has a 
state space $\{ \pm 1 \}^n$ and an {\em energy
function} or {\em Hamiltonian}, $H : \{ \pm 1 \}^n \to \reals$.
The Gibbs measure of the spin system is a probability
distribution assigning to each
state a probability proportional to $e^{-\beta H(\state)}$
where $\beta > 0$ is a parameter called the {\em inverse
temperature}. Application of an {\em external field} 
$\extfld \in \reals^n$ to the spin system is modeled by 
subtracting a linear function from the Hamiltonian, so that it becomes
$H(\state) - \extfld \cdot \state$. The probability of
state $\state$ under the Gibbs measure then becomes
\[
\Pr(\state) = 
 e^{\beta [\extfld \cdot \state - H(\state)] } / Z(\beta,\extfld),
\]
where $Z(\cdot)$ is the {\em partition function}
\[
   Z(\beta,\extfld) = \sum_{\state \in \{\pm 1\}^n} 
                       e^{\beta [\sum_i h_i \sigma_i - H(\state)] } .
\]

Databases $\fdb \in \{0,1\}^n$ are in one-to-one correspondence with states
$\state \in \{\pm 1\}^n$ under the mapping $\sigma_i = (-1)^{\bg_i}$
and its inverse mapping $\bg_i = \frac12 (1-\sigma)$.
Any joint distribution
$\jd = \{0,1\}^n$ has a corresponding Hamiltonian
$H(\state) = - \ln \jd(\fdb)$ whose Gibbs
distribution (at $\beta=1$) equals $\jd$. The positive
affiliation condition is equivalent to requiring
that $H$ is submodular, a property which is expressed
by saying that the spin system is {\em ferromagnetic}.

For a maximally $0$-biased mechanism $\M$ with distinguished
outcome set $S$, the probabilities
$\m(\fdb) = \Pr(\M(\fdb) \in S)$ satisfy
$
  \m(\fdb) \propto \exp(- \sum_{i=1}^n \idpi \bg_i)
      = \exp(- \tfrac{n}{2} + \tfrac12 \sum_i \idpi \sigma_i),
$
so 
$$
  \jd(\fdb) \m(\fdb) \propto e^{- \tfrac{n}{2} + \tfrac12 \sum_i \idpi \sigma_i
                            - H(\state)}.
$$
Application of the mechanism $\M$ is thus analogous to
application of the external field $\frac12 \bm{\idp}$
at inverse temperature 1. 
(The additive constant $-\frac{n}{2}$ in the Hamiltonian
is irrelevant, since the Gibbs measure is unchanged by an
additive shift in the Hamiltonian.) Similarly, applying
a maximally 1-biased mechanism 
is analogous to applying
the external field $- \frac12 \bm{\idp}$
at inverse temperature 1. 

Let $\prat = \frac{\jd(\ba=1)}{\jd(\ba=0)}$ denote the prior 
probability ratio for Athena's bit. 
For the networked
privacy guarantee in \autoref{t-pa}, when
the maximum on the right side of~\eqref{eq:pa} is achieved
by a maximally 0-biased mechanism, we have
\begin{align} \nonumber
\frac{e^{\ndp_a}-\prat}{e^{\ndp_a}+\prat} 
& = 
 \frac{\sum_{\fdb = (0,\pdb)} \jd(\fdb) \m(\fdb) -
                         \sum_{\fdb = (1,\pdb)} \jd(\fdb) \m(\fdb)}
                        {\sum_{\fdb = (0,\pdb)} \jd(\fdb) \m(\fdb) +
                         \sum_{\fdb = (1,\pdb)} \jd(\fdb) \m(\fdb)}
\label{eq:pa2}
= 
\ifextabs
 \frac{\sum_{\state} \sigma_a e^{\bm{\idp} \cdot \state/2 - H(\state)}}
      {\sum_{\state} e^{\bm{\idp} \cdot \state/2 - H(\state)}}
\else
 \frac{\sum_{\state} \sigma_a e^{\bm{\idp} \cdot \state/2 - H(\state)}}
                        {Z(1,\bm{\idp}/2)} 
\fi
  = 
 \expect[\sigma_a \mid \extfld = \tfrac{\bm{\idp}}{2}],
\end{align}
where the operator $\expect[\cdot \mid \extfld = \tfrac{\bm{\idp}}{2}]$ 
denotes the expectation under the Gibbs measure corresponding to 
external field $\extfld = \tfrac{\bm{\idp}}{2}$. A similar calculation
in the case that a maximally 1-biased mechanism maximizes the right side of~\eqref{eq:pa} 
yields the relation $\frac{e^{\ndp_a}-\prat^{-1}}{e^{\ndp_a}-\prat^{-1}} =
\expect[- \sigma_a \mid \extfld = -\tfrac{\bm{\idp}}{2}]$. Combining these two
cases, we arrive at:
\begin{equation} \label{eq:gibbs-ndp}
  \ndp_a = \max \left\{ \ln \prat + \ln \left( 
                                      \frac{1 + \expect[\sigma_a \mid \extfld = \bm{\idp} / 2]}
                                               {1 - \expect[\sigma_a \mid \extfld = \bm{\idp} / 2]}
	                           \right), \;
                                     -\ln \prat  - \ln \left(
                                     \frac{1 + \expect[\sigma_a \mid \extfld = -\bm{\idp}/2]}
                                              {1 - \expect[\sigma_a \mid \extfld = -\bm{\idp}/2]}
                                     \right)
                          \right\}.
\end{equation}
We will refer to $\expect[\sigma_a]$ as the {\em  magnetization
at site $a$}, by analogy with the usual definition of magnetization in
statistical mechanics as the average
$\frac1n \sum_{i=1}^n \expect \left[ \sigma_i \right]$.
Equation~\eqref{eq:gibbs-ndp} thus shows that the \infpriv guarantee
for a positively affiliated distribution is completely determined by the magnetization at site $a$
when an external field of strength $\pm \bm{\idp}/2$ is applied.

\subsubsection{Ising models and phase transitions}
\label{sec:ising}

Let us now apply this circle of ideas to analyze the ``Zeus's family tree''
example from \autoref{sec:intro}.
Represent Zeus and his progeny as the nodes of a
rooted tree, and suppose that the joint distribution of the
individuals' bits is defined by the following
sampling rule: sample the bits in top-down order
(from root to leaves), setting the root's bit to
0 or 1 with equal probability and each other node's
bit equal to the parent's value with probability
$p > \frac12$ and the opposite value otherwise.
This leads to a probability distribution $\jd$ in which
the probability of any $\fdb \in \{0,1\}^{V(T)}$ is
proportional to $p^{a(\fdb)} (1-p)^{b(\fdb)}$ where 
$a(\fdb)$ denotes the number of tree edges whose
endpoints receive the same label, and $b(\fdb)$ is
the number of edges whose endpoints receive
opposite labels. Letting $J = \tanh^{-1}(2p-1)$
so that $\ln(p) = \ln(1-p) + 2 J$, and associating
$\state \in \{\pm 1\}^n$ to $\fdb \in \{0,1\}^n$ via
$\sigma_i = (-1)^{\bg_i}$ as before, we find that
up to an additive constant, 
\ifextabs
$\ln \jd(\fdb) = J \sum_{(i,j) \in E} \sigma_i \sigma_j,$
\else
\[
\ln \jd(\fdb) = J \sum_{(i,j) \in E} \sigma_i \sigma_j,
\]
\fi
where $E$ denotes the edge set of the tree.
Hence, the joint distribution of Zeus's family tree is equivalent
to the Gibbs measure of the Hamiltonian 
$H(\state) = - J \sum_{(i,j) \in E} \sigma_i \sigma_j$.
Models whose Hamiltonian takes this form (for any graph,
not just trees) are known as {\em Ising models} 
(with {\em interaction strength $J$}) and are
among the most widely studied in mathematical physics.

Ising models are known to undergo {\em
phase transitions} as one varies the inverse temperature or
external field. For example, in an infinite two-dimensional
lattice or $\Delta$-regular tree, there is a phenomenon known
as {\em spontaneous magnetization} where the magnetization
does not converge to zero as the external field converges to
zero from above, but this phenomenon only occurs if the
inverse temperature is above a critical value, $\beta_c$,
that is equal to $\ln(1+\sqrt{2})$ for the two-dimensional
lattice and to $\frac12 \ln( 1 + \frac{2}{\Delta-2})$ for the
$\Delta$-regular tree.
This phenomenon of phase transitions has consequences
for \infpriv, as articulated in \cref{t-ising} below. 
To state the theorem it is useful to make the following
definition.

\begin{definition} \label{def:enforces}
Let $\jds$ be a family of joint distributions on
$\{0,1\}^*$, with each distribution $\jd \in \jds$ being
supported on $\bn$ for a specific value $n=n(\jd)$.
For a differential privacy parameter $\idp>0$, let
$\ndp(\idp,\jds)$ denote the supremum,
over all joint distributions
$\jd \in \jds$,
of the \infpriv guarantee corresponding to differential
privacy parameter $\idp$. 
We say that $\ndp$ is 
{\em differentially enforceable} with respect to $\jds$
if there exists $\idp>0$ such that $\ndp \leq \ndp(\idp,\jds)$.
\end{definition}
In other words, to say that $\ndp$ is differentially enforceable 
means that a
regulator can ensure $\ndp$-\infpriv for the individuals participating
in a datasest by mandating that
an analyst must satisfy $\idp$-differential privacy
when releasing the results of an analysis performed on
the dataset.

\begin{theorem} \label{t-ising}
For a family of graphs $\graphs$ and a given $J>0$, let
$\jds$ be the family of Ising models with interaction strength $J$
and zero external field on graphs in $\graphs$.
Then
\begin{enumerate}[a.]
\item \label{sens-cor}
{\bf (Sensitivity to strength of correlations.)}
If $\graphs$ is the set of trees of maximum degree $\Delta=d+1$
and $J = \tanh^{-1} \left( \frac{1-\delta}{d} \right)$ 
for some $\delta > 0$,
then every $\ndp>0$ is differentially enforceable, and in fact
$\ndp(\idp,\jds) = \Theta(\idp/\delta)$ for $0 < \delta < \idp \ll 1$.
On the other hand, if $J > \tanh^{-1} \left( \frac{1}{d} \right)$ then 
the set of all differentially enforceable
$\ndp$ has a strictly positive infimum, $\ndp_{\min}(J,\Delta)$.
\item \label{sens-idp}
{\bf (Sensitivity to differential privacy parameter.)}
For any $0 < \idp_0 < \idp_1$ and any $1 < r < R$, there exists
a joint distribution $\jd$ whose \infpriv guarantee
satisfies $\ndp/\idp < r$ when $\idp=\idp_0$ but 
$\ndp/\idp > R$ when $\idp = \idp_1$.
\end{enumerate}
\end{theorem}
\noindent 
Part (\ref{sens-idp}) is particularly striking 
because it implies, for instance, that when a policy-maker 
contemplates whether to mandate differential
privacy parameter $\idp=0.19$ or $\idp=0.2$, this seemingly
inconsequential decision could determine 
whether the \infpriv guarantee will be 
$\ndp = 0.2$ or $\ndp = 20$. 

\S\ref{sec:bethe} is devoted to
proving the theorem.
The proof combines known facts about phase transitions with some
calculations regarding magnetization of Ising models on
a tree subjected to an external field.

\subsubsection{The Bethe lattice and the proof of \cref{t-ising}}
\label{sec:bethe}

The infinite $\Delta$-regular tree is known in mathematical
physics as the {\em Bethe lattice with coordination number
$\Delta$}. Most of the results stated in \cref{t-ising} can
be derived by analyzing the Ising model on the Bethe lattice
and calculating the magnetization of the root when the 
lattice is subjected to an external field.
Throughout this section, we will use the notation
$\gavg{\sigma_a}$ to denote the expectation of the random
variable $\sigma_a$ under the distribution defined
by the Ising model with interaction strength $J$ at inverse temperature $1$ 
and external field $h$.

\begin{lemma} \label{lem:fixpt}
For given $\jj > 0$, $d \geq 2$, and $h \in \reals$, define a function
$y(x)$ by
\[
    y(x) = e^{2 h} \left( \frac{e^{\jj} x + e^{-\jj}}
                             {e^{\jj} + e^{-\jj} x}
                  \right)^d.
\]
The sequence $x_0,x_1,\ldots$ defined recursively by
$x_0=1$ and $x_{n+1} = y(x_n)$ for $n \geq 0$ converges
to a limit point $x = x(\jj,h)$. This limit point is determined
as follows.
\begin{itemize}
\item If $h=0$ then $x(\jj,h)=1$.
\item If $h>0$ then $x(\jj,h)$ is the unique solution 
of the equation $x=y(x)$ in the interval $(1,\infty)$.
\item If $h<0$ then $x(\jj,h)$ is the unique solution
of the equation $x=y(x)$ in the interval $(0,1)$.
\end{itemize}
The behavior of $x(\jj,h)$ near $h=0$ depends on the
value of $\jj$. 
If $\tanh(\jj) < \frac1d$, then $x(\jj,h)$ varies continuously
with $h$ and $\lim_{h \to 0} x(\jj,h) = 1$.
If $\tanh(\jj) > \frac1d$, then the function $x(\jj,h)$ is 
discontinuous at $h=0$, and it satisfies
$\lim_{h \searrow 0} x(\jj,h) > 1$
and $\lim_{h \nearrow 0} x(\jj,h) < 1$.
\end{lemma}
\begin{proof}
Rewriting the formula for $y(x)$ as 
\[
  y(x) = e^{2 (h + \jj d)} \left[ 
         1 - \frac{e^{\jj} - e^{-3 \jj}}
                  {e^{\jj} + e^{-\jj} x}
         \right]^d
\]
it is clear that for $0 \leq x < \infty$, $y(x)$ is 
continuous and monotonically
increasing in $x$ and takes values between $e^{2 (h - \jj d)}$ and
$e^{2 (h+ \jj d)}$. Since $y$ is monotonic, the sequence
$x_0, x_1, x_2, \ldots$ defined in the lemma must be 
monotonic: if $x_0 \leq x_1$ then an easy induction
establishes that $x_n \leq x_{n+1}$ for all $n$, and 
likewise if $x_0 \geq x_1$ then $x_n \geq x_{n+1}$ for
all $n$. Any monotonic sequence in a closed, bounded interval
must converge to a limit, so the limit point $x(\jj,h)$ 
is well-defined. 

If $h=0$ then a trivial calculation shows that $x_n=1$ for all $n$,
and thus $x(\jj,h)=1$. For $h > 0$ or $h < 0$ we must show 
that $x(\jj,h)$ is the unique solution of $y(x)=x$ in the
interval $(1,\infty)$ or $(0,1)$, respectively. First note
that $y(1) = e^{2 \jj h}$, so the sequence $x_0,x_1,\ldots$
is monotonically increasing when $h>0$ and decreasing when
$h<0$. Thus $x=x(\jj,h)=\lim_{n \to \infty} x_n$ 
belongs to $(1,\infty)$ when $h>0$
and to $(0,1)$ when $h<0$. The continuity
of $y$ implies that
\[
  y(x) = \lim_{n \to \infty} y(x_n) = \lim_{n \to \infty} x_{n+1} = x.
\]
Thus, $x$ satisfies $x = y(x)$. It remains to show that this
equation has a unique solution in $(1,\infty)$ when $h>0$ and
a unique solution in $(0,1)$ when $h<0$.
 
A solution to $x=y(x)$ is also a solution to 
$\ln x - \ln y(x)=0$. The function $g(x) = \ln x - \ln y(x)$
has derivative
\begin{align*}
  g'(x) &= \frac1x - \frac{d e^{\jj}}{e^{\jj} x + e^{-\jj}} 
                  + \frac{d e^{-\jj}}{e^{\jj} + e^{-\jj} x} 
        = \frac1x - \frac{d (e^{2\jj} - e^{-2 \jj})}
                          {x^2 + (e^{2\jj} + e^{-2\jj})x + 1}
\end{align*}
The equation $g'(x)=0$ is equivalent to the quadratic equation
$x^2 - 2 [d \sinh(2 \jj) - \cosh(2 \jj)] x + 1 = 0$.
This has at most two real roots, and if it has any real roots
at all then all roots are real and 
their product is equal to 1. Therefore, it has
at most one root in the interval $(0,1)$ and at most
one root in the interval $(1,\infty)$. Furthermore,
$g'(x)$ is strictly positive at $x=0$ and as $x \to \infty$.
Summarizing 
this discussion, there exist positive numbers $x_0 \leq x_1$ such 
that $x_0 \cdot x_1 = 1$ and the
set $\{x \mid g'(x)>0\}$ intersects 
the intervals $(0,1)$ and $(1,\infty)$ in the
subintervals $(0,x_0)$ and $(x_1,\infty)$, respectively.

Now suppose $h > 0$. The set $\{x \mid x > 1 \mbox{ and } y(x)=x\}$
is non-empty; for example, it contains $x(\jj,h)$. Let $x_{\inf}$
denote the infimum of this set. 
By continuity, $y(x_{\inf})=x_{\inf}$. 
Since $g(x_{\inf})=0$ whereas $g(1) < 0$, we must 
have $g'(z) > 0$ for some $z$ in the interval $(1,x_{\inf})$. 
Recalling the number $x_1$ defined in the previous paragraph,
we must have $x_1 < z < x_{\inf}$. Consequently $g'$ is strictly positive 
throughout the interval $(x_{\inf},\infty)$, implying that there
are no other solutions of $g(x)=0$ in that interval. Thus,
$x_{\inf}$ is the unique solution of $y(x_{\inf})=x_{\inf}$ in
$(1,\infty)$, and $x(\jj,h) = x_{\inf}$. When $h<0$ an
analogous argument using 
$x_{\sup} = \sup \{ x \mid x < 1 \mbox{ and } y(x)=x\}$
proves that $x(\jj,h) = x_{\sup}$ is the unique
solution of $y(x)=x$ in the interval $(0,1)$.

To analyze the behavior of $x(\jj,h)$
near $h=0$, it is useful to first analyze the
zero set of $g'(x)$. Recall that $g'(x)=0$ if
and only if 
$x^2 - 2 [d \sinh(2 \jj) - \cosh(2 \jj)] x + 1 = 0$.
The discriminant test tells us that this quadratic equation 
has zero, one, or two real roots according to whether 
$d \sinh(2 \jj) - \cosh(2 \jj) - 1$ is less than, equal to, or 
greater than 0. Using the identities 
$\sinh(2 \jj) = 2 \sinh(\jj) \cosh(\jj) = 2 \cosh^2(\jj) \tanh(\jj)$
and $\cosh(2 \jj) = 2 \cosh^2(\jj) - 1$
we find that 
$d \sinh(2 \jj) - \cosh(2 \jj) - 1 = 
 2 \cosh^2(\jj) [ d \tanh(\jj) - 1 ]$.
So when $\tanh(\jj) > \frac1d$,
$g'(x) > 0$ for all $x$ and the equation
$y(x)=x$ has the unique solution $x(\jj,h)$.
Implicit differentiation, applied to the equation
$x(\jj,h) = y(x(\jj,h))$, yields:
\begin{equation} \label{eq:bethe.1}
\tfrac{\partial x}{\partial h} = \tfrac{\partial y}{\partial h} + 
                 \tfrac{\partial y}{\partial x} \cdot \tfrac{\partial x}{\partial h}
\end{equation}
which can be rearranged to yield
\begin{equation} \label{eq:bethe.2}
\tfrac{\partial x}{\partial h} = \tfrac{\partial y / \partial h}{1 - \partial y / \partial x}.
\end{equation}
The function $y$ is $C^\infty$ in the region $x>0$, and
at $h=0, x=1$ we have
\begin{align} 
\nonumber
\tfrac{\partial y}{\partial x} & = d \tanh(\jj) \\
\nonumber
\tfrac{\partial y}{\partial h} & = 2  \\
\label{eq:impl-diff}
\tfrac{\partial x}{\partial h} & = \tfrac{2}{1 - d \tanh(\jj)},
\end{align}
so when $\tanh(\jj) < \frac1d$ 
the implicit function theorem implies that $x(\jj,h)$ is a
differentiable, increasing function of $h$ in a neighborhood
of $h=0$.

When $\tanh(\jj) > \frac1d$ and $h=0$, we have 
$g(1)=0$ and $g'(1) < 0$, so for some sufficiently small $\delta>0$
we have $g(1+\delta) < 0, \, g(1-\delta) > 0$. On the other
hand, the fact that $\ln y(x)$ is bounded between 
$- 2 \jj d$ and $2 \jj d$ implies that 
$g(x) = \ln x - \ln y(x)$ tends to $-\infty$ as $x \to 0$
and to $\infty$ as $x \to \infty$. The intermediate value
theorem implies that there exist $x_+ \in (1+\delta, \infty)$
and $x_- \in (0,1-\delta)$ such that $g(x_+) = g(x_-)=0$.
In fact, the equation $g(x)=0$ can have at most three 
solutions since $g'(x)=0$ has only two solutions. So, the
entire solution set of $g(x)=0$ is $\{x_-, 1, x_+\}$.
Denote the function $y(x)$ in the case $h=0$ by
$y_0(x)$, to distinguish it from the case of general $h$;
similarly define $g_0(x) = \ln x - \ln y_0(x)$. Note that
$g_0(x) \leq 0$ when $1 \leq x \leq x_+$, so 
$x \leq y_0(x)$ on that interval. When $h>0$
we have $y(x) > y_0(x)$ for all $x$, hence
$y(x) > x$ for $1 \leq x \leq x_+$. As 
$x(\jj,h)$ is the unique solution of $y(x)=x$
in the interval $(1,\infty)$ it follows that $x(\jj,h) > x_+$.
On the other hand, for any
$\delta > 0$, we have $g_0(x_+ \,\! + \,\! \delta) > 0$ and hence, 
for sufficiently small $h > 0$, we also have 
$g(x_+ \,\! + \,\! \delta) > 0$. Since $g(x_+) < 0$ 
and $g(x(\jj,h))=0$, the intermediate
value theorem implies $x(\jj,h)$ belongs
to the interval $(x_+,x_+ \,\! + \,\! \delta)$ for all
sufficiently small $h>0$. In other words,
$\lim_{h \searrow 0} x(\jj,h) = x_+$.
The analogous argument for $h<0$ proves that 
$x(\jj,h) < x_-$ and that 
$\lim_{h \nearrow 0} x(\jj,h) = x_-$. 
\end{proof}

%
%

\begin{lemma} \label{lem:subtree}
If $T$ is a subtree of $T'$ and $a$ is any node of $T$,
let $\gavg{\sigma_a}_T$ and $\gavg{\sigma_a}_{T'}$ denote the
expectation of $\sigma_a$ in the Ising models on $T$ and $T'$,
respectively, with interaction
strength $\jj>0$. For $h > 0$ we have $\gavg{\sigma_a}_{T'} \geq \gavg{\sigma_a}_T$
while for $h < 0$ we have $\gavg{\sigma_a}_{T'} \leq \gavg{\sigma_a}_T$.
\end{lemma}
\begin{proof}
It suffices to prove the lemma in the case that $h>0$ 
(since the $h<0$ case is symmetric under exchanging the 
signs $+1$ and $-1$) and that $T$ is obtained 
from $T'$ by deleting a single leaf node, $b$. The lemma then follows
by induction, since any subtree can be obtained from a tree by 
successively deleting leaves.

Let $c$ denote the parent of $b$ in $T'$, \ie, assume that 
$(b,c)$ is the unique edge of $T'$ containing $b$.
For state $\state \in \{\pm 1\}^{V(T)}$, 
$(\state,+)$ and $(\state,-)$ denote the states in $\{\pm 1\}^{V(T')}$
obtained by setting $\sigma_b=+1$ or $\sigma_b=-1$, respectively, 
while keeping the spin at every node of $T$ the same. 
If 
$H(\state) = -\jj \sum_{(i,j) \in E(T)} \sigma_i \sigma_j$
is the Hamiltonian of the Ising model on $T$, then the
Hamiltonian of the Ising model on $T'$ is given by
\begin{align*}
H'(\state,+) &= -\jj \sigma_c + H(\state) \\
H'(\state,-) &= \jj \sigma_c + H(\state).
\end{align*}
Thus, the partition functions $Z(\beta,\extfld), Z'(\beta,\extfld)$
of $T,T'$ respectively satisfy
\begin{align*}
Z(\beta,\extfld) &= \sum_{\state} e^{\beta[\extfld \cdot \state - H(\state)]} \\
Z'(\beta,\extfld) &= \sum_{\state} e^{\beta[\extfld \cdot (\state,+) -
                                           H'(\state,+)]} +
                                  e^{\beta[\extfld \cdot (\state,-) -
                                           H'(\state,-)]} \\
  &= 
  2 \sum_{\state} \cosh(\beta[h + J \sigma_c]) \,
            e^{\beta[\extfld \cdot \state - H(\state)]} .
\end{align*}
Furthermore, we have
\begin{align*}
\gavg{\sigma_a}_T &= \frac{1}{Z(\beta,\extfld)}
           \sum_{\state} \sigma_a \, e^{\beta[\extfld \cdot \state - H(\state)]} \\
\gavg{\sigma_a}_{T'} &= \frac{1}{Z'(\beta,\extfld)}
           \sum_{\state} \sigma_a  \,
                                  e^{\beta[\extfld \cdot (\state,+) -
                                           H'(\state,+)]} +
                                  e^{\beta[\extfld \cdot (\state,-) -
                                           H'(\state,-)]} \\
  &=
  \frac{2}{Z'(\beta,\extfld)}
  \sum_{\state} \sigma_a \, \cosh(\beta[h + J \sigma_c]) \,
            e^{\beta[\extfld \cdot \state - H(\state)]} .
\end{align*}
Associating to each $\fdb \in \bn$ a state $\state(\fdb) \in \{\pm 1\}^n$
via $\sigma_i = (-1)^{\dbi}$ as before, we find that  
the logarithm of the function $\fdb \mapsto 
e^{\beta[\extfld \cdot \state(\fdb) - H(\state(\fdb))]}$
is supermodular. Furthermore, the functions
$\fdb \mapsto (-1)^{\dba}$ and
$\fdb \mapsto 2 \cosh(\beta[h+(-1)^{\bg_c} J])$
are both monotonically decreasing. Thus, we may apply the
FKG inequality (\cref{lem:fkg}) to conclude that
\begin{align*} 
\lefteqn{
\left[
\sum_{\state} e^{\beta[\extfld \cdot \state - H(\state)]} \right]
\; \left[
2 \sum_{\state} \sigma_a \cosh(\beta[h + J \sigma_c]) 
            e^{\beta[\extfld \cdot \state - H(\state)]} \right]} \\
& \qquad \qquad \geq
\left[  2 \sum_{\state} \cosh(\beta[h + J \sigma_c]) \,
            e^{\beta[\extfld \cdot \state - H(\state)]} \right]
\; \left[
           \sum_{\state} \sigma_a e^{\beta[\extfld \cdot \state - H(\state)]} \right].
\end{align*}
Dividing both sides by $Z(\beta,\extfld) \cdot Z'(\beta,\extfld)$, we obtain
the inequality asserted in the lemma.
\end{proof}

\begin{lemma} \label{lem:bethe}
If $T$ is a finite tree of maximum degree $\Delta=d+1$, $a$ is any node of $T$,
and $\gavg{\sigma_a}$ denotes the expectation of $\sigma_a$ in the
Ising model on $T$ with interaction strength $\jj$, inverse temperature $1$, 
and external field $h>0$,
then 
\begin{equation} \label{eq:bethe}
  \ln \left( \frac{1 + \gavg{\sigma_a}}{1 - \gavg{\sigma_a}} \right)
  <
  \frac{d+1}{d} \ln x(\jj,h) - \frac{2 h}{d} .
\end{equation}
The difference between the left and right sides converges to 
zero as the distance from $a$ to the nearest node of degree
less than $\Delta$ tends
to infinity.
\end{lemma}
\begin{proof}
Define a sequence of rooted trees $T_0, T_1, \ldots$ recursively,
by stating that $T_0$ is a single node and $T_{n+1}$ consists of a 
root joined to $d = \Delta-1$ children, each of whom is the root
of a copy of $T_n$. Also define a sequence of trees $T^*_0, T^*_1, \ldots$,
by stating that $T^*_0 = T_0$ while for $n>0$, $T^*_n$ consists of a 
root joined to $\Delta$ children, each of whom is the root
of a copy of $T_{n-1}$. (In other words, $T^*_n$ is like $T_n$,
with the root modified to have $\Delta$  instead of $\Delta-1$
children.)

If $T$ is any tree of maximum degree $\Delta$ containing
a node labeled $a$, then let $r$ denote the distance 
from $a$ to the nearest
node of degree less than $\Delta$, and let $s$ denote 
the distance from $a$ to the farthest leaf. We can 
embed $T^*_{r}$ as a subtree of $T$ rooted at $a$, and
we can embed $T$ as a subtree of $T^*_{s}$ with $a$ at the
root. Applying \autoref{lem:subtree},
$$
  \gavg{\sigma_a}_{T^*_r} \leq
  \gavg{\sigma_a}_{T} \leq
  \gavg{\sigma_a}_{T^*_s}.
$$
To complete the proof we will show that 
$\ln \left( \frac{1+\gavg{\sigma_a}_{T^*_n}}{1-\gavg{\sigma_a}_{T^*_n}} 
\right)$ converges to $\frac{d+1}{d} \ln x(\jj,h) - \frac{2 h}{d}$ 
(from below) as $n \to \infty$. 

For any tree $T$ with root node $k$, let 
\begin{align*}
  Z^+(T) &= \sum_{\state : \sigma_k=+1} e^{\jj \sum_{i,j} \sigma_i \sigma_j
            +  h \sum_i \sigma_i}  \\
  Z^-(T) &= \sum_{\state : \sigma_k=-1} e^{\jj \sum_{i,j} \sigma_i \sigma_j
            + h \sum_i \sigma_i} 
\end{align*}
We have $\gavg{\sigma_k}_T = \frac{Z^+(T)-Z^-(T)}{Z^+(T) + Z^-(T)}$,
so $\frac{1+\gavg{\sigma_k}_T}{1-\gavg{\sigma_k}_T} = \frac{Z^+(T)}{Z^-(T)}$.
For the tree $T_n$ defined in the preceding paragraph, the quantity
$x_n = Z^+(T_{n})/Z^-(T_{n})$ satisfies the recurrence
\[
  x_{n+1} = \frac{e^{h} (e^{\jj} Z^+(T_{n}) + e^{-\jj} Z^-(T_n))^d}
                 {e^{-h} (e^{-\jj} Z^+(T_n) + e^{\jj} Z^-(T_n))^d}
         = e^{2 h} \left(
                    \frac{e^{\jj} x_n + e^{-\jj}}{e^{\jj} + e^{-\jj} x_n}
                  \right)^d
         = y(x_n)
\]
where the function $y(\cdot)$ is defined as in 
\cref{lem:fixpt}. Applying the conclusion of 
that lemma, we find that $x_n$ increases with $n$
and $x_n \to x(\jj,h)$ from below
as $n \to \infty$. Finally, for the quantity
$x^*_n = Z^+(T^*_n)/Z^-(T^*_n)$ we have
\[
  x^*_{n} = \frac{e^{h} (e^{\jj} Z^+(T_{n-1}) 
                               + e^{-\jj} Z^-(T_{n-1}))^{d+1}}
                 {e^{-h} (e^{-\jj} Z^+(T_{n-1}) 
                               + e^{\jj} Z^-(T_n))^{d+1}}
         = e^{2 h} \left(
                    \frac{e^{\jj} x_{n-1} + e^{-\jj}}{e^{\jj} 
                              + e^{-\jj} x_{n-1}}
                  \right)^{d+1}
         = e^{-2 h / d} x_n^{(d+1)/d}.
\]
Finally,
\[
 \ln \left( \frac{1 + \gavg{\sigma_a}}{1 - \gavg{\sigma_a}} \right)
    = 
 \ln x^*_n
    = 
 \tfrac{d+1}{d} \ln(x_n) - \tfrac{2 h}{d}.
\]
The lemma follows because
$\ln(x_n)$ increases with $n$
and converges to $\ln x(\jj,h)$ 
from below as 
$n \to \infty$.
\end{proof}

\begin{corollary} \label{cor:bethe-ndp}
Let $\jds$ denote the family of Ising models
with interaction strength $\jj$ and
zero external field on trees of maximum
degree $\Delta$. For any $\idp > 0$, 
\begin{equation} \label{eq:bethe-ndp}
  \ndp(\idp,\jds) = \tfrac{\Delta}{\Delta-1} 
                    \ln x \! \left(\jj,\tfrac{\idp}{2} \right) 
                     - \tfrac{\idp}{\Delta-1}.
\end{equation}
\end{corollary}
\begin{proof}
For the Ising model with zero external field, the
joint distribution $\jd$ is symmetric with respect
to flipping each bit of the database $\fdb$. This
implies two simplifications in the formula for
\infpriv, Equation~\eqref{eq:gibbs-ndp}. First,
the odds ratio $\rho = \frac{\jd(\dba=1)}{\jd(\dba=0)}$
is equal to 1. Second, both terms in the maximum 
on the right-hand side of the equation are equal,
so $\ndpa$ is equal to $\ln \left( 
\frac{1 + \gavg{\sigma_a}}{1 - \gavg{\sigma_a}} \right)$,
where the $\gavg \cdot$ denotes averaging over the 
Gibbs measure (at inverse temperature 1) 
of the Ising model with interaction 
strength $\jj$ and external field $\idp/2$. 
Applying
\cref{lem:bethe} we obtain~\eqref{eq:bethe-ndp}
as a direct consequence.
\end{proof}

\begin{proof}[Proof of \cref{t-ising}.]
For part (\ref{sens-cor}) of the theorem, 
\cref{cor:bethe-ndp} justifies focusing our attention
on the function $\ln x \! \left(\jj,h \right)$ where
$h = \frac{\idp}{2}$ and $\idp>0$ varies. 
In particular, when $\tanh(\jj) > \frac{1}{d}$,
we have 
\[ 
  \lim_{\idp \searrow 0} \ndp(\idp,\jds) 
  = \tfrac{\Delta}{\Delta-1} \lim_{h \searrow 0} \ln x(\jj,h/2)
  > 0
\]
by \cref{lem:fixpt}.
This implies that the set of differentially enforceable $\ndp$
has a strictly positive infimum, as claimed in
part (\ref{sens-cor}) of \cref{t-ising}.

When $\tanh(\jj) = \frac{1-\delta}{d}$, 
\cref{eq:impl-diff} for the partial
derivative $\frac{\partial x}{\partial h}$
implies that 
$\tfrac{\partial x}{\partial h} = \frac{2}{\delta}$ at $h=0,x=1$.
We now find that
\begin{align*}
\left. \tfrac{d \ndp(\idp,\jds)}{d \idp} \right|_{\idp=0} &=
\left( \tfrac{\Delta}{\Delta-1} \right) \tfrac{\partial}{\partial \idp}
\left[ \ln x \! \left( \jj, \tfrac{\idp}{2} \right) \right]_{\idp=0}
- \left( \tfrac{1}{\Delta-1}  \right)
\\
&=
\tfrac{\Delta}{\Delta-1} \cdot \tfrac{1}{x(\jj,0)}
\cdot \tfrac{1}{2} \cdot \left[ \tfrac{\partial x}{\partial h} \right]_{h=0}
- \left( \tfrac{1}{\Delta-1} \right) \\
&=
\tfrac{\Delta}{\Delta-1} \cdot 1
\cdot \tfrac{1}{2} \cdot \tfrac{2}{\delta}
- \left( \tfrac{1}{\Delta-1} \right) \\
&= 
\left( \tfrac{\Delta}{\Delta-1} \right) \tfrac{1}{\delta} - 
\left( \tfrac{1}{\Delta-1} \right)  > \frac{1}{\delta}.
\end{align*}
Thus, for sufficiently small $\idp>0$, 
we have $\ndp(\idp,\jds) > \idp/\delta$, 
which completes the proof of part (\ref{sens-cor}) of the theorem.

To prove part (\ref{sens-idp}) we consider rooted $d$-ary
trees for some fixed $d \geq 2$. As in the proof of 
\cref{lem:bethe} let $T_n$ denote the complete rooted $d$-ary
tree of depth $n$, with root node denoted by $a$. For 
$J > 0, h \in \reals$ define
\[
    w_n(J,h) = \ln \left( \frac{1 + \gavg{\sigma_a}}{1-\gavg{\sigma_a}} \right)
\]
where $\gavg{\sigma_a}$ denotes the expectation of $\sigma_a$
under the Ising model on $T_n$ with interaction strength $J$ and
external field $h$. In the proof of \cref{lem:bethe} we denoted
$\exp(w_n(J,h))$ by $x_n$ and proved that $x_n \to x(J,h)$ 
from below as $n \to \infty$.

Now consider an adversary whose prior $\jd$ is the Ising model
with interaction strength $J$ and external field $h_0$. Note that
the prior odds ratio $\rho = \frac{\jd(x_a=1)}{\jd(x_a=0)}$
satisfies
\[
  \ln \rho = \ln \left( \frac{\jd(\sigma_a = -1)}{\jd(\sigma_a = 1)} \right)
              = \ln \left( \frac{1 - \gavg{\sigma_a}}{1+\gavg{\sigma_a}} \right)
              = -w_n(J,h_0).
\]
Substituting this into \cref{eq:gibbs-ndp}, we see that 
for a given differential
privacy parameter $\eps$
the corresponding \infpriv guarantee is 
\begin{equation} \label{eq:sens-idp.1}
  \ndp(\eps) = \max \{ w_n(J,h_0 + \tfrac{\eps}{2}) - w_n(J,h_0),
            \; w_n(J,h_0) - w_n(J, h_0- \tfrac{\eps}{2}) \}.
\end{equation}
To prove \cref{t-ising}(\ref{sens-idp}) consider 
any $0 < \eps_0 < \eps_1$ and $1 < r < R$. 
In setting up the adversary's prior, choose a value of $h_0$ that 
satisfies $2h_0 - \eps_1 < 0 < 2h_0 - \eps_0$. 
We aim to show that for all sufficiently large $J$ and all
$n > n_0(J)$, we have
$\ndp(\eps_0) < r \cdot \eps_0$ but 
$\ndp(\eps_1) > R \cdot \eps_1$. 

Let $w(J,h) = \lim_{n \to \infty} w_n(J,h) = \ln x(J,h)$.
To prove that $\ndp(\eps_0) < r \cdot \eps_0$ but 
$\ndp(\eps_1) > R \cdot \eps_1$ for all sufficiently 
large $n$, it is sufficient to prove that 
\begin{align}
\label{eq:sens-ndp.2}
    \tfrac{w(J, h_0+ \frac12 \eps_0) - w(J,h_0)}{\eps_0} & < r \\
\label{eq:sens-ndp.3}
    \tfrac{w(J,h_0) - w(J, h_0 - \frac12 \eps_0)}{\eps_0} & < r \\
\label{eq:sens-ndp.4}
  \tfrac{w(J,h_0) - w(J,  h_0 - \frac12 \eps_1)}{\eps_1} & > R.
\end{align}
To prove~\eqref{eq:sens-ndp.2}-\eqref{eq:sens-ndp.3} 
we will show that 
$| \partial w / \partial h |$ is bounded above by $2r$ on
the interval $[h_0 - \frac12 \eps_0, h_0+\frac12 \eps_0]$
and apply the mean value theorem. Since $w(J,h) = \ln x(J,h)$
we have
\begin{equation} \label{eq:sens-ndp.5}
  \frac{\partial w}{\partial h} = 
  \frac{\partial x / \partial h}{x} = 
  \frac{(\partial y / \partial h) / x}{1 - \partial y / \partial x} = 
  \frac{2}{1 - \partial y / \partial x},
\end{equation}
where we have used \cref{eq:impl-diff} and the facts
that $\partial y / \partial h = 2 y$ and that
$y(x(J,h)) = x(J,h)$. Now, recalling the definition
of $y(x)$ in \cref{lem:fixpt}, we differentiate with
respect to $x$ and find that
\begin{align}
\nonumber
  \frac{\partial y}{\partial x} &= 
  e^{2h} \cdot d \left(
    \frac{e^{J} x + e^{-J}}{e^{J} + e^{-J} x}
  \right)^{d-1} \cdot
  \left( \frac{1 - e^{-4J}}{(e^J + e^{-J} x)^2} \right) \\
\nonumber
&=
  y(x) \cdot 
  \left( \frac{ d \cdot \left( 1 - e^{-4J} \right)}{ \left( e^J x + e^{-J} \right)
     \left( e^J + e^{-J} x \right)} \right) \\
\label{eq:sens-ndp.6}
& <
  \frac{y(x)}{x} \cdot \frac{d}{e^{2J}}.
\end{align}
Since $y(x)/x = 1$ when $x  = x(J,h)$, we may combine
\eqref{eq:sens-ndp.5} with~\eqref{eq:sens-ndp.6} to
conclude that whenever $J$ is large enough that 
$d e^{-2J} < 1 - \frac1r$, then the value of
$\partial y / \partial x$ at $x(J,h)$ is less than 
$1 - 1/r$ for all $h>0$, and consequently
$\partial w / \partial h$ is bounded above uniformly
by $r$, as desired.

Finally, to prove~\eqref{eq:sens-ndp.4} we note that
for $x = e^{2h+Jd}$ we have
\begin{equation} \label{eq:sens-ndp.7} 
  \frac{y(x)}{x} = e^{2h} \left( \frac{e^{J} x + e^{-J}}{e^{J} + e^{-J} x}
  \right)^{d} e^{-2h-Jd} 
  = \left( \frac{e^{2h+Jd} + e^{-2J}}{e^{2h + Jd - J} + e^J} \right)^{d} > 1
\end{equation}
for $J$ sufficiently large. Recalling from the proof
of \cref{lem:fixpt} that for $h>0$ we have $y(x) > x$ when $1 < x < x(J,h)$
and $y(x) < x$ when $x > x(J,h)$, we see that
$x(J,h) > e^{2h+Jd}$ provided that $h>0$ and $J$ is sufficiently
large. An analogous argument applying the $h<0$ case
of \cref{lem:fixpt} shows that $x(J,h) < e^{2h - Jd}$ 
for $h<0$ and $J$ sufficiently large. Recalling that
$w(J,h) = \ln x(J,h)$ and that 
$h_0 > 0 > h_0 - \eps_1/2$ we find that
\begin{align*}
  w(J,h_0) &> 2h_0 + Jd > Jd \\
  w(J,h_0 - \eps_1/2) &< 2 h_0 - \eps_1 - Jd < -Jd \\
  \tfrac{w(J,h_0) - w(J,  h_0 - \frac12 \eps_1)}{\eps_1} &>
  \tfrac{2Jd}{\eps_1} > R
\end{align*}
provided $J$ is sufficiently large. This establishes~\eqref{eq:sens-ndp.7} 
and concludes the proof of \cref{t-ising}(\ref{sens-idp}).
\end{proof}

\section{Appendix to \S\ref{sec:bounded}: Bounded Affiliation Distributions}
\label{app:bounded}

This appendix contains a full proof
of \autoref{t-bm}. The proof requires
developing a theory of ``multiplicative estimates''
that is the multiplicative analogue of the notion of 
``estimate'' used by \citet{Dob70}, \citet{Follmer82}, and \citet{Kunsch82} 
in their proofs of the so-called Dobrushin Comparison Theorem.
We define multiplicative estimates and 
build up the necessary machinery for dealing with them
in \autoref{sec:preserve}.
Then, in \autoref{sec:dobrushin} we 
prove \autoref{t-bm}.

\subsection{Multiplicative estimates}
\label{sec:preserve}

Let $S$ be any set of
potential outcomes of the mechanism $\M$ such that
 $\Pr(\M(\fdb) \in S)>0$. Let $\cd{1}$ denote the 
conditional distribution on databases $\fdb \in X^n$,
given that $\M(\fdb) \in S$, and let
$\cd{2}$ denote the unconditional distribution $\jd$,
respectively.

For $a \in \{1,2\}$ and for any function $f : X^n
\to \reals$, let $\cd{a}(f)$
denote the expected value of $f$ under distribution
$\cd{a}$. Also define the Lipschitz constants
\begin{equation} \label{eq:lip}
\lip_i(f) = \max 
            \{ f(\fdb) - f(\adj{\fdb}) 
               \; \mid \; \fdb \sim_i \adj{\fdb} \}.
\end{equation}
Let us say that a vector 
$\bm{\estim} = (\estim_{i})$ is a {\em multiplicative estimate}
if for every function $f : X^n \to \reals_+$ we have
\begin{equation} \label{eq:est}
| \ln \cd{1}(f) - \ln \cd{2}(f) | \leq \sum_{i=1}^{n} \estim_i 
          \lip_i(\ln f).
\end{equation}

This section is devoted to proving some basic facts
about multiplicative estimates that underpin the proof of
\autoref{t-bm}. 
To start, we need the following lemma.

\begin{lemma} \label{lem:fg}
Consider a probability space with two functions $A,B$ taking values
in the positive real numbers.
If $(\sup A)/(\inf A) \leq e^{2a}$ and 
$(\sup B)/(\inf B) \leq e^{2b}$ then
\begin{equation} \label{eq:lfg}
  \frac{\expect[AB]}{\expect[A] \, \expect[B]} 
    \leq 
  1 + \frac{(e^{2a}-1)(e^{2b}-1)}{(e^a + e^b)^2} 
    \leq
  e^{ab}.
\end{equation}
\end{lemma}
\begin{proof}
The hypotheses and conclusion of the lemma are invariant
under rescaling each of $A$ and $B$, so we may assume without
loss of generality that $A$ is supported in the interval
$[e^{-a},e^{a}]$ and that $B$ is supported in the interval
$[e^{-b},e^{b}]$. At each sample point $\omega$, the following
two equations hold:
\begin{align}
\label{eq:fg.1}
   \begin{bmatrix}
      \tfrac{e^{a} - A(\omega)}{e^a - e^{-a}} &
      \tfrac{A(\omega) - e^{-a}}{e^a - e^{-a}}
   \end{bmatrix} \;
   \begin{bmatrix} 1 & e^{-a} \\ 1 & e^{a} \end{bmatrix} 
 &=
   \begin{bmatrix} 1 & A(\omega) \end{bmatrix} \\
\label{eq:fg.2}
   \begin{bmatrix}
      \tfrac{e^{b} - B(\omega)}{e^a - e^{-b}} &
      \tfrac{B(\omega) - e^{-b}}{e^a - e^{-b}} 
   \end{bmatrix} \;
   \begin{bmatrix} 1 & e^{-b} \\ 1 & e^{b} \end{bmatrix} 
 &=
   \begin{bmatrix} 1 & B(\omega) \end{bmatrix} \\
\end{align}
Therefore, if we define the matrix-valued random variable
\[ 
    M(\omega) =    \begin{bmatrix}
      \tfrac{e^{a} - A(\omega)}{e^a - e^{-a}} &
      \tfrac{A(\omega) - e^{-a}}{e^a - e^{-a}} 
   \end{bmatrix}^{\transpose} \;
   \begin{bmatrix}
      \tfrac{e^{b} - B(\omega)}{e^b - e^{-b}} &
      \tfrac{B(\omega) - e^{-b}}{e^b - e^{-b}}
   \end{bmatrix}
\]
we have 
\[
     \begin{bmatrix} 1 & 1 \\ e^{-a} & e^{a} \end{bmatrix} \; 
     M(\omega) \;
     \begin{bmatrix} 1 & e^{-b} \\ 1 & e^{b} \end{bmatrix} =
     \begin{bmatrix} 1 \\ A(\omega) \end{bmatrix} \;
     \begin{bmatrix} 1 & B(\omega) \end{bmatrix} =
     \begin{bmatrix} 1 & B(\omega) \\ A(\omega) & A(\omega) B(\omega)
     \end{bmatrix} 
     .
\]
Integrating over $\omega$ we obtain
\begin{align}
\nonumber
  \begin{bmatrix} 1 & 1 \\ e^{-a} & e^{a} \end{bmatrix} \; 
   \expect[M] \;
  \begin{bmatrix} 1 & e^{-b} \\ 1 & e^{b} \end{bmatrix} 
 &=
  \begin{bmatrix} 1 & \expect[B] \\ \expect[A] & \expect[AB] \end{bmatrix} \\
\nonumber
  \expect[M] &= 
  \begin{bmatrix} 1 & 1 \\ e^{-a} & e^{a} \end{bmatrix}^{-1} \; 
  \begin{bmatrix} 1 & \expect[B] \\ \expect[A] & \expect[AB] \end{bmatrix} \;
  \begin{bmatrix} 1 & e^{-b} \\ 1 & e^{b} \end{bmatrix}^{-1} \\
\nonumber
  \left( e^{a} - e^{-a} \right) 
  \left( e^{b} - e^{-b} \right)
  \expect[M] 
 &=
  \begin{bmatrix} e^{a} & -1 \\ -e^{-a} & 1 \end{bmatrix} \; 
  \begin{bmatrix} 1 & \expect[B] \\ \expect[A] & \expect[AB] \end{bmatrix} \;
  \begin{bmatrix} e^{b} & -e^{-b} \\ -1 & 1 \end{bmatrix}  \\
\nonumber
 \lefteqn{\clap{\hspace*{0.35\textwidth} $=
  \begin{bmatrix} 
    e^{a+b} - e^a \expect[B] - e^b \expect[A] + \expect[AB] &
    -e^{a-b} + e^a \expect[B] + e^{-b} \expect[A] - \expect[AB] \\
    -e^{b-a} + e^{-a} \expect[B] + e^b \expect[A] - \expect[AB] &
    e^{-a-b} - e^{-a} \expect[B] - e^{-b} \expect[A] + \expect[AB]
  \end{bmatrix}.$}} 
\end{align}
Each entry of the matrix on the left side is non-negative, hence
the entries on the right side are non-negative as well. This tells
us that
\begin{align} 
\nonumber
  \expect[AB] & \leq \min \left\{
      e^{a} \expect[B] + e^{-b} \expect[A] - e^{a-b},
      e^{b} \expect[A] + e^{-a} \expect[B] - e^{b-a} 
      \right\} \\
\label{eq:fg.3}
  & = 
      \expect[A] \expect[B] +
      \min \left\{
        (e^{a} - \expect[A]) (\expect[B] - e^{-b}),
        (e^{b} - \expect[B]) (\expect[A] - e^{-a})
      \right\}.
\end{align}
Letting 
\[
\alpha = \tfrac{1}{\expect[A]}, \quad
\beta = \tfrac{1}{\expect[B]},
\]
we can multiply both sides of~\eqref{eq:fg.3} by $\alpha \beta$ to obtain
\begin{equation} \label{eq:fg.4}
\frac{\expect[AB]}{\expect[A] \, \expect[B]} \leq
1 + \min \{ (e^{a} \alpha - 1)(1 - e^{-b} \beta),
            (e^{b} \beta - 1)(1 - e^{-a} \alpha) \}.
\end{equation}
Denote the right side of~\eqref{eq:fg.4} by $G(\alpha,\beta)$.
We aim to find the maximum value of $G(\alpha,\beta)$ as
$(\alpha,\beta)$ ranges over the rectangle
$[e^{-a}, e^a] \times [e^{-b}, e^b]$. 
Note that $G \equiv 1$ on the boundary of this
rectangle, whereas $G > 1$ on the interior of the
rectangle. At any point of the interior where
$ (e^{a} \alpha - 1)(1 - e^{-b} \beta) >
            (e^{b} \beta - 1)(1 - e^{-a} \alpha)$
we have $\frac{\partial G}{\partial \beta} = e^b (1 - e^{-a} \alpha) > 0$,
and similarly at any point of the interior where
$ (e^{a} \alpha - 1)(1 - e^{-b} \beta) <
            (e^{b} \beta - 1)(1 - e^{-a} \alpha)$
we have $\frac{\partial G}{\partial \alpha} > 0$.
Therefore if $(\alpha,\beta)$ is a global maximum of $G$ we must have
$(e^{a} \alpha - 1)(1 - e^{-b} \beta) =
            (e^{b} \beta - 1)(1 - e^{-a} \alpha)$.
Let
\begin{equation} \label{eq:fg.5}
 r = \frac{e^a \alpha - 1}{1 - e^{-a} \alpha} = 
       \frac{e^b \beta - 1}{1 - e^{-b} \beta}.
\end{equation}
A manipulation using~\eqref{eq:fg.5} yields
\begin{equation} \label{eq:fg.6}
(e^{a} \alpha - 1)(1 - e^{-b} \beta) = 
\left( e^{2a}-1 \right) \left( e^{2b} - 1 \right) 
\frac{r}{\left( r + e^{2a} \right) \left( r + e^{2b} \right)}
\end{equation}
and by setting the derivative of the right side to zero we find
that it is maximized at $r=e^{a+b}$, when it equates to
$(e^{2a-1})(e^{2b}-1)(e^a+e^b)^{-2}$. Therefore,
\[
  \frac{\expect[AB]}{\expect[A] \, \expect[B]} \leq
  1 + \frac{\left( e^{2a}-1 \right) \left( e^{2b} - 1 \right) }
           {\left( e^a + e^b \right)^2},
\]
which establishes the first inequality in~\eqref{eq:lfg}.
The prove the second inequality, we
consider how $1 + (e^{2a}-1)(e^{2b}-1)(e^a+e^b)^{-2}$ varies
as we vary $a$ and $b$ while holding their product fixed
at some value, $x^2$.
To begin we compute the gradient of $1 + (e^{2a}-1)(e^{2b}-1)(e^a+e^b)^{-2}$.
\begin{align*}
  \nabla \left[ 1 + (e^{2a}-1)(e^{2b}-1)(e^a+e^b)^{-2} \right] &= 
  \begin{bmatrix}
     \left( 
       \frac{2e^{2a}(e^{2b}-1)}{(e^a+e^b)^2} - 
       \frac{2e^a (e^{2a}-1) (e^{2b} - 1)}{(e^a+e^b)^3} 
     \right) &
     \left( 
       \frac{2e^{2b}(e^{2a}-1)}{(e^a+e^b)^2} - 
       \frac{2e^b (e^{2a}-1) (e^{2b} - 1)}{(e^a+e^b)^3} 
     \right)
  \end{bmatrix}   
\\
  &= 
  \frac{2(e^{a+b}-1)}{(e^a+e^b)^3}
  \begin{bmatrix}
    e^a(e^{2b}-1) & e^b(e^{2a}-1)
  \end{bmatrix}
\\
  &=
  \frac{8 \left(1 - e^{-a-b} \right)}{(e^a+e^b)^3}
  \begin{bmatrix}
    \sinh b & \sinh a
  \end{bmatrix}
\end{align*}
Parameterizing the curve $ab = x^2$ by $a(t) = xt, b(t) = x/t$,
we have $\dot{a}(t) = a/t$ and $\dot{b}(t) = -b/t$, so
\begin{align*}
  \frac{d}{dt} \left[ 1 + (e^{2a}-1)(e^{2b}-1)(e^a+e^b)^{-2} \right] &=
  \frac{8 \left(1 - e^{-a-b} \right)}{(e^a+e^b)^3}
  \begin{bmatrix}
    \sinh b & \sinh a
  \end{bmatrix} \,
  \begin{bmatrix}
    a/t \\ -b/t
  \end{bmatrix} 
\\
  &=
  \frac{8 ab \left(1 - e^{-a-b} \right)}{(e^a+e^b)^3 t}
  \left( \frac{\sinh b}{b} - \frac{\sinh a}{a} \right).
\end{align*}
From the Taylor series $\frac{\sinh y}{y} = 
\sum_{i=0}^{\infty} \frac{1}{(2i+1)!} y^{2i}$
we see that $\frac{\sinh y}{y}$ is an increasing
function of $y \geq 0$, so along the curve $a(t) = xt, b(t)=x/t$,
the function $1 + (e^{2a}-1)(e^{2b}-1)(e^a+e^b)^{-2}$ increases
when $a < b$ (corresponding to $t<1$) and decreases when 
$a > b$ (corresponding to $t>1$), reaching its maximum when
$t=1$ and 
$a = b = x$. Hence
\begin{equation} \label{eq:fg.7}
1 + \frac{\left( e^{2a}-1 \right) \left( e^{2b}-1 \right)}
         {\left( e^a+e^b \right)^{2} } \leq
1 + \left( \frac{e^{2x}-1}{2 e^x} \right)^2 =
1 + \sinh^2 x = 
\cosh^2 x .
\end{equation}
Finally, by comparing Taylor series coefficients
we can see that $\cosh x \leq e^{x^2/2}$ for all $x \geq 0$,
and squaring both sides of this relation 
we obtain
\begin{equation} \label{eq:fg.8}
\cosh^2 x \leq e^{x^2} = e^{ab}.
\end{equation}
The second inequality in~\eqref{eq:lfg} follows by 
combining~\eqref{eq:fg.7} with~\eqref{eq:fg.8}.
\end{proof}

\begin{lemma} \label{lem:Ti}
If $\bm{\estim}$ is a multiplicative estimate, then for any $i \in 
\{1,\ldots,n\}$, the vector
$T_i(\bm{\estim})$ defined by
\begin{equation} \label{eq:Ti}
  (T_i(\bm{\estim}))_\ell = \begin{cases} 
    \estim_\ell & \mbox{if $\ell \neq i$} \\
    \tfrac{\idp_i}{2} + 
    \sum_{j=1}^n \infmul_{ij} \estim_j 
    & \mbox{if $\ell = i$}
  \end{cases}
\end{equation}
is also a multiplicative
estimate.
\end{lemma}
\begin{proof}
For a distribution $\cd{}$ on $X^{n}$, a database 
$\fdb \in X^{n}$, and an individual $i$, let 
$\cd{}(\cdot \mid \fdb_{-i})$ denote the conditional
distribution of $x_i$ given $\fdb_{-i}$. In other words,
$\cd{}(\cdot \mid \fdb_{-i})$ is the probability distribution
on $X$ given by 
\begin{equation} \label{eq:condprob}
  \cd{}(x \mid \fdb_{-i}) = \frac{\cd{}(x,\fdb_{-i})}
                                 {\sum_{y \in X} \cd{}(y,\fdb_{-i})}.
\end{equation}
Letting $W$ denote the vector space of real-valued functions on $X^n$,
we define an averaging operator $\avgi : W \to W$ 
which maps a function $f$ to the function 
$$\avgi f(\fdb) = \sum_{y \in X} f(y,\fdb_{-i}) \, \cd{}(y \mid \fdb_{-i}).$$
Equivalently, $\avgi f$ is the unique function 
satisfying:
\begin{enumerate}
\item The value $\avgi f(\fdb)$ depends only on $\fdb_{-i}$.
\item For any other function $g$ whose value $g(\fdb)$ depends only on
$\fdb_{-i}$, we have 
\begin{equation} \label{eq:condexp}
\cd{}(f g) = \cd{}((\avgi f) g).
\end{equation}
\end{enumerate}
Note that $\cd{}(f) = \cd{}(\avgi f)$, as can be seen from 
applying~\eqref{eq:condexp} to the constant function $g(\fdb)=1$.

For the distributions $\cd{1}$ and $\cd{2}$ defined earlier, 
let us denote the corresponding averaging operators by 
$\avgi^1$ and $\avgi^2$. Using the identities 
$\cd{1}(f) = \cd{1}(\avgi^1 f)$ and 
$\cd{2}(f) = \cd{2}(\avgi^2 f)$, we find that
\begin{equation} \label{eq:kunsch.1}
| \ln \cd{1}(f) - \ln \cd{2}(f) | \leq
| \ln \cd{1}(\avgi^1 f) - \ln \cd{1}(\avgi^2 f)| +
| \ln \cd{1}(\avgi^2 f) - \ln \cd{2}(\avgi^2 f)|.
\end{equation}
We bound the two terms on the right side separately.
For the first term, we write
\begin{equation} \label{eq:ab1}
  \left| \ln \left(
                 \frac{\cd{1}(\avgi^1 f)}{\cd{1}(\avgi^2 f)} 
             \right) \right | =
  \left| \ln \left(
                \frac{\sum_{\fdb} \cd{1}(\fdb) \avgi^1 f(\fdb)}
                     {\sum_{\fdb} \cd{1}(\fdb) \avgi^2 f(\fdb)}
             \right) \right| \leq
  \max_{\fdb} \left| \ln \left(
                            \frac{\avgi^1 f(\fdb)}{\avgi^2 f(\fdb)} 
                        \right) \right|.
\end{equation}
For any particular $\fdb$, we can bound the ratio
$\frac{\avgi^1 f(\fdb)}{\avgi^2 f(\fdb)}$ from above
using Lemma~\ref{lem:fg} applied to the probability
space $X$, under the distribution $\cd{2}(\cdot \mid \fdb_{-i})$.
Letting $A(x) = f(x,\fdb_{-i}),
B(x) = \cd{1}(x \mid \fdb_{-i})/\cd{2}(x \mid \fdb_{-i})$
we have 
\begin{align} 
  \expect[A(x) B(x)] &= \avgi^1 f(\fdb) \\
  \expect[A(x)]      &= \avgi^2 f(\fdb) \\
  \expect[B(x)]      &= 1 \\
  \tfrac{\sup A}{\inf A} &\leq e^{\lip_i(\ln f)} \\
\label{eq:ab2}
  \tfrac{\sup B}{\inf B} &\leq e^{2 \idp_i}.
\end{align}
The first four of these relations are straightforward,
and the last requires some justification. In the following
calculation we use the 
operator $\Pr(\cdot)$ to denote probabilities of events 
in the sample space where $\fdb$ is sampled from the 
original joint distribution $\jd$, and randomized
mechanism $\M$ is applied to $\fdb$. Starting from the 
definitions of $\cd{1}$ and $\cd{2}$, an application of
Bayes' Law yields the following calculation.
\begin{align*}
  \cd{1}(x \mid \fdb_{-i}) &= 
    \frac{\Pr(\fdb = (x,\fdb_{-i}) \mid \M(\fdb) \in S)}
         {\Pr(\fdb \in X \times \{\fdb_{-i}\} \mid \M(\fdb) \in S)} \\
  &= \frac{\Pr(\M(\fdb) \in S \mid \fdb = (x,\fdb_{-i}))}
          {\Pr(\M(\fdb) \in S \mid \fdb \in  X \times \{\fdb_{-i}\})} 
     \; \cdot \;
     \frac{\Pr(\fdb = (x,\fdb_{-i}))}{\Pr(\fdb \in X \times \{\fdb_{-i}\})}.
\end{align*}
The first factor on the right-hand side is between $e^{-\idp_i}$ 
and $e^{\idp_i}$,
while the second factor is equal to $\cd{2}(x \mid \fdb_{-i})$.
This completes the proof of~\eqref{eq:ab2}. By combining
Lemma~\ref{lem:fg} with the contents of~\eqref{eq:ab1}-\eqref{eq:ab2}
we obtain the bound
\begin{equation} \label{eq:kunsch.2}
| \ln \cd{1}(\avgi^1 f) - \ln \cd{1}(\avgi^2 f)| \leq
\tfrac12 \idp_i \lip_i(\ln f).
\end{equation}

To bound the second term in~\eqref{eq:kunsch.1}, 
we will make use of the 
following inequality, a multiplicative analogue 
of inequality (3.5) in \citep{Gross79}.
\begin{equation} \label{eq:gross}
\forall i,j \; \; 
\lip_j(\ln \avgi^{2} f) \leq \begin{cases}
  0 & \mbox{if $i=j$} \\
  \lip_j(\ln f) +  \infmul_{ij} \, \lip_i(\ln f) & \mbox{if $i \neq j$}.
\end{cases}
\end{equation}
The validity of~\eqref{eq:gross} is evident when $i=j$, since
the value $\avgi^2 f(\fdb)$ does not depend on $\bg_i$.
To prove~\eqref{eq:gross} when $i \neq j$, we use the definition of
$\lip_j(\cdot)$ to choose $\fdb,\adb \in X^{n}$
such that $\fdb \sim_j \adb$ and 
\begin{align} \nonumber
\lip_j(\ln \avgi^{2} f) &= \ln \avgi^{2} f(\adb) - \ln \avgi^{2} f(\fdb) \\
\nonumber
    &= \ln \left( \sum_{x \in X} \cd{2}(x \mid \adb_{-i}) f(x,\adb_{-i}) \right)
    - \ln \left( \sum_{x \in X} \cd{2}(x \mid \fdb_{-i}) f(x,\fdb_{-i}) \right) \\
\label{eq:gross.2}
    &\leq
       \left| 
        \ln \left( \frac{\sum_{x \in X} \cd{2}(x \mid \adb_{-i}) f(x,\adb_{-i})}
                        {\sum_{x \in X} \cd{2}(x \mid \adb_{-i}) f(x,\fdb_{-i})}
            \right) \right|
       +
       \left| 
        \ln \left( \frac{\sum_{x \in X} \cd{2}(x \mid \adb_{-i}) f(x,\fdb_{-i})}
                        {\sum_{x \in X} \cd{2}(x \mid \fdb_{-i}) f(x,\fdb_{-i})}
                  \right) \right|.
\end{align}
Using the fact that $\frac{f(x,\adb_{-i})}{f(x,\fdb_{-i})} \leq e^{\lip_j(\ln f)}$
for all $x \in X$, we see that the first term on the right side 
of~\eqref{eq:gross.2} is bounded
above by $\lip_j(\ln f)$. To bound the second term, we again make use of
Lemma~\ref{lem:fg}, this time substituting $A(x) = f(x,\fdb_{-i})$ and
$B(x) = \frac{\cd{2}(x \mid \adb_{-i})}{\cd{2}(x \mid \fdb_{-i})}$. Taking
expectations under
the probability distribution $\cd{2}(x \mid \fdb_{-i})$ we have
\begin{align*}
  \expect[A(x) B(x)] &= \sum_{x \in X} \cd{2}(x \mid \adb_{-i}) f(x,\fdb_{-i}) \\
  \expect[A(x)     ] &= \sum_{x \in X} \cd{2}(x \mid \fdb_{-i}) f(x,\fdb_{-i}) \\
  \expect[B(x)     ] &= 1 \\
  \tfrac{\sup A}{\inf A} &\leq e^{\lip_i(\ln f)} \\
  \tfrac{\sup B}{\inf B} &\leq e^{4 \infmul_{ij}},
\end{align*}
where the last line is justified by observing that the 
definition of $\infmul_{ij}$ ensures that $\sup B \leq e^{2 \infmul_{ij}}$
and $\inf B \geq e^{-2 \infmul_{ij}}$.
An application of Lemma~\ref{lem:fg} immediately implies
that the second term on the right side of~\eqref{eq:gross.2} is
bounded above by $ \infmul_{ij} \lip_i(\ln f)$.
This completes the proof of~\eqref{eq:gross}.

Now, our hypothesis that $\bm{\estim}$ is a multiplicative
estimate implies, by definition, that
\begin{equation} \label{eq:kunsch.3}
| \ln \cd{1}(\avgi^2 f) - \ln \cd{2}(\avgi^2 f) | \leq
\sum_{j=1}^{n} \estim_j \lip_j(\ln \avgi^2 f) \leq
\sum_{j \neq i} \estim_j \lip_j(\ln f) + 
\sum_{j=1}^{n} \estim_j \infmul_{ij} \, \lip_i(\ln f)
\end{equation}
where we have applied~\eqref{eq:gross} to derive the second
inequality. Combining~\eqref{eq:kunsch.1},~\eqref{eq:kunsch.2},
and~\eqref{eq:kunsch.3} we now have
\[
  | \ln \cd{1}(f) - \ln \cd{2}(f) | \leq 
  \sum_{j \neq i} \estim_j \lip_j(\ln f) +
  \left( \tfrac12 \idp_i + \sum_{j=1}^n \infmul_{ij} \estim_j \right) 
                \lip_i(\ln f).
\]
\end{proof}


\begin{lemma} \label{lem:T}
If $\bm{\estim}$ is a multiplicative estimate, then 
the vector
$T(\bm{\estim}) = \tfrac{1}{2n} \bm{\idp} + 
\left( 1 - \tfrac{1}{n} \right) \bm{\estim} +
\tfrac{1}{n} \infmulmtx \bm{\estim}$
is also a multiplicative
estimate.
\end{lemma}
\begin{proof}
Averaging~\eqref{eq:Ti} over $i=1,\ldots,n$, we find that
$
  T(\bm{\estim}) = \tfrac1n T_i(\bm{\estim}).
$
The lemma follows because the set of multiplicative estimates
is closed under convex combinations, as is evident from the
definition of a multiplicative estimate.
\end{proof}

\begin{lemma} \label{lem:invmtx}
If 
the influence matrix $\infmulmtx$ has spectral norm
strictly less than 1, then $\frac12 \invmtx \bm{\idp}$
is an estimate.
\end{lemma}
\begin{proof}
To begin, let us prove that 
the set of multiplicative estimates is non-empty; in fact, it
contains the vector $\bm{1} = (1,\ldots,1)$.
To see this, consider any $f : X^n = \reals_+$ and choose
$\fdb \in \argmax(f), \adb \in \argmin(f)$. Define a sequence
$(\fdb^{(k)})_{k=0}^n$ by the formula
\[
  \fdb^{(k)}_i = \begin{cases}
    \bg_i & \mbox{if $i > k$} \\
    \bg'_i & \mbox{if $i \leq k$}
  \end{cases}.
\]      
Note that $\fdb^{(0)} = \fdb$, $\fdb^{(n)} = \adb$,
and $\fdb^{(k-1)} \sim_k \fdb^{(k)}$ for $k=1,\ldots,n$.
Therefore,
\begin{align*}
| \ln \cd{1}(f) - \ln \cd{2}(f) | \leq
| \max(\ln f) - \min(\ln f) | & =
| \ln f(\fdb) - \ln f(\adb) | \\  & \leq
\sum_{k=1}^n | \ln f(\fdb^{(k-1)}) - \ln f(\fdb^{(k)}) | \leq
\sum_{k=1}^n \lip_k(\ln f),
\end{align*}
so $\bm{1}$ is an estimate as claimed.

Now let $\cmtx = \left( 1 - \frac1n \right) I + \frac1n \infmulmtx$.
Applying \autoref{lem:T} inductively, each element of the sequence
\[
  T^m(\bm{1}) = \frac{1}{2n} \left( \sum_{k=0}^{m-1} \cmtx^k 
                                  \right) \bm{\idp} + 
                     \cmtx^m \bm{1}
\]
is a multiplicative estimate. If $\| \infmulmtx \| < 1$ (where 
$\| \cdot \|$ denotes spectral norm) then $\cmtx$ also has spectral
norm less than 1 because it is a convex combination of $\infmulmtx$
and $I$. This implies that the sequence 
$(T^m(\bm{\estim}))_{m=0}^{\infty}$ converges to
$\frac{1}{2n} (I - \cmtx)^{-1} \bm{\idp}$.
Now,
\[
  (I - \cmtx)^{-1} = (\tfrac1n I - \tfrac1n \infmulmtx )^{-1} =
  n (I - \infmulmtx)^{-1},
\]
so the sequence $(T^m(\bm{\estim}))_{m=0}^{\infty}$
converges to $\frac12 \invmtx \bm{\idp}$. 
The proof concludes with the  observation that a
limit of multiplicative estimates is again a multiplicative
estimate.
\end{proof}

\subsection{Proof of \autoref{t-bm}}
\label{sec:dobrushin}
Let us begin by restating \autoref{t-bm}.

\begin{theorem}
Suppose that the joint distribution $\fdb$ has a multiplicative
influence matrix $\infmulmtx$ whose spectral norm is strictly
less than 1. Let $\invmulmtx = (\invmulentry_{ij})$ denote the
matrix inverse of $I-\infmulmtx$. Then for any mechanism with
individual privacy 
parameters $\bm{\idp} = (\idpi)$, 
the networked differential privacy guarantee
satisfies
\begin{equation} \label{eq:ndp-dobr}
\forall i \;\; 
\ndp_i \leq 2 \sum_{j=1}^{n} \invmulentry_{ij} \idp_j.
\end{equation}
If the matrix of multiplicative influences 
satisfies 
\begin{equation} \label{eq:ba-cond}
\forall i \;\;
\sum_{j=1}^{n} \infmul_{ij} \idp_j \leq (1-\delta) \idpi
\end{equation}
for some $\delta>0$, then 
\begin{equation} \label{eq:ndp-ba}
\forall i \;\; \ndp_i \leq 2 \idpi / \delta.
\end{equation}
\end{theorem}
\begin{proof}

Above, in \cref{lem:invmtx}, we proved that
$\frac12 \invmtx \bm{\idp}$ is a multiplicative estimate.
In other
words, for any $f : X^{n} \to \reals_{+}$ it holds that
\begin{equation} \label{eq:fest}
\left| \ln \cd{1}(f) - \ln \cd{2}(f) \right| \leq 
\tfrac12 \sum_{i,j=1}^{n} \invmtx_{ij} \idp_j \lip_i(\ln f).
\end{equation}
To prove~\eqref{eq:ndp-dobr}, we are required to show the 
following:
if $\val_0,\val_1$ are any two distinct elements of $X$ such that 
$\Pr(\dbi = \val_0)$ and $\Pr(\dbi = \val_1)$ are both positive, then
\begin{equation} \label{eq:dobr-to-prove.1}
\left| \ln \left( \frac{
         \Pr(\dbi=\val_1 \mid \M(\fdb) \in S) \, / \, 
         \Pr(\dbi=\val_0 \mid \M(\fdb) \in S) 
       }{
         \Pr(\dbi=\val_1) \, / \, \Pr(\dbi=\val_0)
       }        
       \right)
\right|
    \leq
2 \sum_{j=1}^n \invmtx_{ij} \idp_j .
\end{equation}
We will do this by setting $f$ and $g$ to be the indicator
functions of the events $\dbi = \val_0$ and
$\dbi = \val_1$, respectively. Then~\eqref{eq:dobr-to-prove.1}
can be rewritten in the form
\begin{equation} \label{eq:dobr-to-prove.2}
\left| \ln \left( \frac{
          \cd{1}(f) / \cd{1}(g)
       }{
          \cd{2}(f) / \cd{2}(g)
       }  
       \right)
\right|      
    \leq
2 \sum_{j=1}^n \invmtx_{ij} \idp_j .
\end{equation}
If the Lipschitz constants of $f$ and $g$ satisfied
$\lip_j(\ln f) = \lip_j(\ln g) = 0$ for $j \neq i$ 
and $\lip_i(\ln f), \lip_i(\ln g) \leq 1$, 
then~\eqref{eq:dobr-to-prove.2} would follow
immediately by applying~\eqref{eq:fest} to $f$
and $g$ separately. Instead 
$\lip_i(f) = \lip_i(g) = \infty$
so we will have to be more indirect,
applying~\eqref{eq:fest} to $\avgf$ and $\avgg$
where $\avg$ is an averaging operator designed
to smooth out $f$ and $g$, thereby improving
their Lipschitz constants. Specifically, define
\begin{align*}
  \avgf(\fdb) = \cd{2}(\val_0,\fdb_{-i}), \qquad
  \avgg(\fdb) = \cd{2}(\val_1,\fdb_{-i}).
\end{align*}
It is 
useful to describe $\avgf$ and $\avgg$ in terms of the 
following sampling process: generate a coupled pair of
samples $(\fdb',\fdb'')$ by sampling $\fdb'$ from $\cd{2}$,
then resampling $\bg''_i$ from the conditional distribution 
$\cd{2}(\cdot \mid \fdb'_{-i})$, and then assembling
the database $\fdb'' = (\bg''_i, \fdb'_{-i})$. Then
$\avgf(\fdb)$ is the conditional probability that $\bg''_i = \val_0$
given that $\fdb' = \fdb$, and $\avgg$ is defined similarly
using $\val_1$ instead of $\val_0$. An important observation
is that the distribution of $(\fdb',\fdb'')$ is exchangeable,
\ie $(\fdb',\fdb'')$ and $(\fdb'',\fdb')$ have the same
probability. From this observation we can immediately
conclude that 
\begin{equation} \label{eq:dobr.1}
  \cd{2}(\avgf) = \cd{2}(f), \quad
  \cd{2}(\avgg) = \cd{2}(g),
\end{equation}
because $\cd{2}(\avgf)$ is the probability that $\bg''_i = \val_0$
whereas $\cd{2}(f)$ is the probability that $\bg'_i = \val_0$, and
similarly for $g$ and $\val_1$. Our strategy for 
proving~\eqref{eq:dobr-to-prove.2} will be to bound
the left side using
\begin{align} 
\nonumber
\left| \ln \left( \frac{
          \cd{1}(f) / \cd{1}(g)
       }{
          \cd{2}(f) / \cd{2}(g)
       }  
       \right)
\right|   
    & =
\left| \ln \left( \frac{
          \cd{1}(f) / \cd{1}(g)
       }{
          \cd{2}(\avgf) / \cd{2}(\avgg)
       }  
       \right)
\right| \\ 
\label{eq:dobr.2}
    & \leq
\left| \ln \left( \frac{
          \cd{1}(f) / \cd{1}(g)
       }{
          \cd{1}(\avgf) / \cd{1}(\avgg)
       }  
       \right)
\right| +
\left| \ln \left( \frac{
          \cd{1}(\avgf) / \cd{1}(\avgg)
       }{
          \cd{2}(\avgf) / \cd{2}(\avgg)
       }  
       \right)
\right| 
\end{align}
and to bound the two terms on the last line separately.
For the second term we will use~\eqref{eq:fest} applied
to $\avgf$ and $\avgg$ separately. This
requires us to bound the Lipschitz constants
$\lip_k(\ln \avgf)$ and $\lip_k(\ln \avgg)$.
Since $\avgf(\fdb)$ and $\avgg(\fdb)$ do not
depend on $\bg_i$, it is immediate that 
$\lip_f(\ln \avgf)=\lip_i(\ln \avgg)=0$.
For $k \neq i$, the definition of the multiplicative influence
parameter $\infmul_{ik}$ leads to the bounds
\begin{equation} \label{eq:dobr.3}
\lip_k(\ln \avgf), \, \lip_k(\ln \avgg) \leq 2 \infmul_{ik}.
\end{equation}
Note that~\eqref{eq:dobr.3} also holds when $k=i$ 
since $\infmul_{ii} = 0$.
Applying~\eqref{eq:fest} to $\avgf$ and $\avgg$ yields the bound
\begin{align} \nonumber
\left| \ln \left( \frac{
          \cd{1}(\avgf) / \cd{1}(\avgg)
       }{
          \cd{2}(\avgf) / \cd{2}(\avgg)
       }  
       \right)
\right| & \leq 
\left| \ln \cd{1}(\avgf) - \ln \cd{2}(\avgf) \right| +
\left| \ln \cd{1}(\avgg) - \ln \cd{2}(\avgg) \right| \\
\label{eq:dobr.4}
& \leq
2 \cdot \frac12 \cdot \sum_{k,j=1}^n (2 \infmul_{ik}) \invmtx_{kj} \idp_j 
= 2 (\infmulmtx \invmtx \bm{\idp})_i
\end{align}
Recalling that $\invmtx = (I-\infmulmtx)^{-1}$, we have
$(I - \infmulmtx) \invmtx = I$ and hence
$\infmulmtx \invmtx = \invmtx - I$. Thus, we can rewrite~\eqref{eq:dobr.4}
as
\begin{equation} \label{eq:dobr.5}
\left| \ln \left( \frac{
          \cd{1}(\avgf) / \cd{1}(\avgg)
       }{
          \cd{2}(\avgf) / \cd{2}(\avgg)
       }  
       \right)
\right|  \leq
2 \sum_{j=1}^n \invmtx_{ij} \idp_j \; - \; 2 \idp_i.
\end{equation}
Now we turn to bounding the first term
in~\eqref{eq:dobr.2}. Letting $\outc$ denote
the random variable representing the mechanism's
outcome, $\M(\fdb)$. Bayes' Law tells us that
\[
  \cd{1}(\fdb) = \frac{\cd{2}(\fdb) \Pr(\outc \in S \given \fdb)}
                      {\Pr(\outc \in S)}.
\]
Therefore,
\begin{align*}
\frac{\cd{1}(\avgg)}{\cd{1}(g)} =
\frac{\sum_{\fdb} \cd{2}(\val_1 \mid \fdb_{-i}) \cd{1}(\fdb)}
     {\sum_{\fdb} g(\fdb) \cd{1}(\fdb)} 
&=
\frac{\sum_{\fdb} \cd{2}(\val_1 \mid \fdb_{-i}) \cd{2}(\fdb) 
      \Pr(\outc \in S \given \fdb)}
     {\sum_{\fdb} g(\fdb) \cd{2}(\fdb) \Pr(\outc \in S \given \fdb)} \\
&=
\frac{\sum_{\fdb}  \cd{2}(\val_1 \mid \fdb_{-i}) \cd{2}(\fdb) 
      \Pr(\outc \in S \given \fdb)} 
     {\sum_{\fdb_{-i}} \cd{2}(\val_1,\fdb_{-i}) 
      \Pr(\outc \in S \given (\val_1,\fdb_{-i}))} \\
&=
\frac{\sum_{\fdb_{-i}} \cd{2}(\val_1 \mid \fdb_{-i})
      \sum_{\val \in \vals} \cd{2}(\val, \fdb_{-i})
          \Pr(\outc \in S \given (\val,\fdb_{-i}))} 
      {\sum_{\fdb_{-i}} \cd{2}(\val_1 \mid \fdb_{-i})
      \sum_{\val \in \vals} \cd{2}(\val, \fdb_{-i})
          \Pr(\outc \in S \given (\val_1,\fdb_{-i}))} .
\end{align*}
The right side lies between $e^{-\idpi}$ and $e^{\idpi}$
because each ratio 
$\frac{\Pr(\outc \in S \given (\val,\fdb_{-i}))}
      {\Pr(\outc \in S \given (\val_1,\fdb_{-i}))}$
lies between $e^{-\idpi}$ and $e^{\idpi}$. Thus,
\begin{equation} \label{eq:dobr.6}
\left| \ln \left( \frac{\cd{1}(\avgg)}{\cd{1}(g)} \right) \right|
  \leq \idpi 
\end{equation}
Similarly,
\begin{equation} \label{eq:dobr.7}
\left| \ln \left( \frac{\cd{1}(f)}{\cd{1}(\avgf)} \right) \right|
  \leq \idpi .
\end{equation}
Combining~\eqref{eq:dobr.6} with~\eqref{eq:dobr.7} yields the bound
\begin{equation} \label{eq:dobr.8}
\left| \ln \left( \frac{
          \cd{1}(f) / \cd{1}(g)
       }{
          \cd{1}(\avgf) / \cd{1}(\avgg)
       }  
       \right)
\right| \leq 2 \idpi.
\end{equation}
Combining~\eqref{eq:dobr.8} with~\eqref{eq:dobr.5} we obtain
the bound~\eqref{eq:dobr-to-prove.2}, which finishes the
proof of the first inequality in the theorem statement,
namely~\eqref{eq:ndp-dobr}.

To prove inequality~\eqref{eq:ndp-ba}, we use the partial
ordering on vectors defined by $\bm{a} \preceq \bm{b}$ if 
and only if $a_i \leq b_i$ for all $i$. 
The matrix $\infmulmtx$ has non-negative entries,
so it preserves this ordering: if $\bm{a} \preceq \bm{b}$ then 
$
   \forall i \; \sum_{j} \infmul_{ij} a_j \leq \sum_{j} \infmul_{ij} b_j
$
and hence $\infmulmtx \bm{a} \preceq \infmulmtx{b}$. 
Rewriting the relation~\eqref{eq:ba-cond} in the form
 $\infmulmtx \bm{\idp} \preceq (1-\delta) \bm{\idp}$
and applying induction, we find that for all $n \geq 0$,
 $\infmulmtx^n \bm{\idp} \preceq (1-\delta)^n \bm{\idp}$.
Summing over $n$ yields
\[
 \invmtx \bm{\idp} = \sum_{n=0}^{\infty} \infmulmtx^n \bm{\idp} \preceq
 \sum_{n=0}^{\infty} (1-\delta)^n \bm{\idp} = \tfrac{1}{\delta} \bm{\idp}
\]
which, when combined with~\eqref{eq:ndp-dobr}, yields~\eqref{eq:ndp-ba}. 
\end{proof}

\end{document}